\documentclass[journal]{IEEEtran}
\usepackage{amsmath,amsfonts}
\usepackage{algorithmic}
\usepackage{algorithm}
\usepackage{array}
\usepackage[caption=false,font=normalsize,labelfont=sf,textfont=sf]{subfig}
\usepackage{textcomp}
\usepackage{stfloats}
\usepackage{url}
\usepackage{verbatim}
\usepackage{graphicx}
\usepackage{cite}
\usepackage{float} 
\usepackage[table,xcdraw]{xcolor} 
\usepackage{bbm} 
\usepackage{amssymb} 
\usepackage{amsfonts} 
\usepackage{tikz} 
\usepackage{amsthm} 
\ifCLASSOPTIONcompsoc 
\usepackage[caption=false,font=normalsize,labelfont=sf,textfont=sf]{subfig}
\else
\usepackage[caption=false,font=normalsize]{subfig}
\fi

\newtheorem{theorem}{Theorem}
\newtheorem{assumption}{Assumption}
\newtheorem{lemma}{Lemma}
\newtheorem{corollary}{Corollary}[lemma]

\newcommand\numberthis{\addtocounter{equation}{1}\tag{\theequation}} 

\newcounter{MYtempeqncnt} 
\def\CommentsColor{black}

\begin{document}

\title{Cellular Load Dependent Sleep Control for Energy Efficient HetNets with Non-Uniform User Distributions}

\author{Martin Willame,~\IEEEmembership{Member,~IEEE,} Charles Wiame,~\IEEEmembership{Member,~IEEE,} Jérôme Louveaux,~\IEEEmembership{Fellow,~IEEE,}\\ Claude Oestges,~\IEEEmembership{Fellow,~IEEE,} and Luc Vandendorpe,~\IEEEmembership{Fellow,~IEEE,}
\thanks{}} 

\markboth{}
{Willame \MakeLowercase{\textit{et al.}}: Cellular Load Dependent Power Control for Energy Efficient HetNets with Non-Uniform User Distributions}


\maketitle

\begin{abstract}
This study proposes a novel stochastic geometry framework analyzing power control strategies in spatially correlated network topologies. Heterogeneous networks are studied, with users modeled via the superposition of homogeneous and Poisson cluster processes. First, a new expression approaching the distribution of the number of users per base station is provided. This distribution defines the load associated with each Voronoï cell, capturing non-uniformities in user locations and correlation to BSs positions. \textcolor{\CommentsColor}{The power allocation is adjusted based on this load, allowing BSs to enter sleep mode when their activity falls below a defined threshold.} Furthermore, the propagation model features millimeter wave transmission characteristics and directional beamforming. Considering these aspects, revisited definitions of coverage probability, spectral efficiency, and energy efficiency are proposed. Tractable expressions for these metrics are derived and validated using Monte-Carlo simulations. \textcolor{\CommentsColor}{Asymptotic expressions are also proposed, providing further understanding on the influence of the system parameters.} Our numerical results finally analyze the impact of the sleep control on the performance and display the optimal strategies in terms of energy efficiency.
\end{abstract}

\begin{IEEEkeywords}
Sleep strategy, Power control, HetNet, Poisson cluster process, Energy efficiency, Millimeter wave, Stochastic geometry
\end{IEEEkeywords}

\section{Introduction}
\IEEEPARstart{T}{he} development of new cellular technologies raises questions about the potential increase in total energy consumption of mobile networks. This growing need for energy could lead to a higher carbon footprint, resulting in potential impacts on the environment \cite{6736745}. In order to address these aspects, a growing interest towards Globally Resource optimized Energy Efficient Network (GREEN) communications has been observed \cite{7939957}. This field of research focuses on resource allocation methods to optimize the Energy Efficiency (EE) of cellular networks. Since the most energy-intensive equipments in cellular mobile networks are Base Stations (BSs) \cite{6514953}, sleep strategies for BSs appear as an attractive way to reduce energy consumption \cite{7041163}. One possible manner to evaluate the benefits of these sleep techniques at a large scale level is Stochastic Geometry (SG). This mathematical field enables to study and characterize the properties of random spatial patterns \cite{PPP_GEN}. As studied in \cite{6042301,6620915}, locations of BSs and User Equipments (UEs) in actual networks can be shown to form such patterns. By abstracting these locations by points in the Euclidean space, SG provides a general framework for analyzing via closed form expressions the performance of cellular networks.

\textcolor{\CommentsColor}{The aim of this paper is to establish a framework for quantifying the EE benefits of sleep control strategies. To this end, closed-form expressions for the performance of cellular networks are derived using SG tools and analysed through simulations.}
\subsection{Related Works} 
Driven by the observation that the traffic demand varies in time and space, the use of on/off sleeping mode has recently been studied in \cite{6502479,9676414,6566900,7036048,6575091,6774407,7136461,7185349,7248320,8458429,6567876,7057551,7881136,7475871,7060678,8372969,8633858,6684701,7792582}. The characteristics and features of the system models analyzed in these works are summarized in Table \ref{tab:SotA}. 
Random Sleeping (RS) strategies have first been considered in the literature for their simplicity. RS is usually modeled as a Bernoulli trial, assuming that each BS is active with probability $q$ and sleeps with probability $1-q$. Reducing the density of active BSs enables to cut down the power consumption as well as the interference between neighboring BSs. However, this comes at a price of a lower average received power per user. As RS does not capture the global network state, network-aware sleeping strategies gained interest as well. Such Strategic Sleeping (SS) captures the impact of the load of each BS. The framework introduced in \cite{6502479} aims at studying the benefits of SS. BSs are asleep when their activity level is below a fixed threshold, and awake otherwise. The benefit of this sleep strategy has been illustrated and compared to a random strategy. This activity level can be defined in different ways, such as the downlink rate or the number of served users. This last definition is used in \cite{9676414} to examine the evolution of the coverage probability in Homogeneous Networks (HomNets). An approximated expression of this probability is also derived for Heterogeneous Networks (HetNets) in the same study.

HetNets model the densification of current networks with different types of cells. However, most works only consider independent BSs and UEs locations, which is not a realistic assumption for current capacity driven networks \cite{6171996}. In \cite{6684701}, Markov theory is used to derive traffic load dependent sleep modes in closed access networks. The UEs are distributed as a Matern core Poisson Cluster Process (PCP). The framework introduced in \cite{6171996} aims at studying Hetnets with non-uniform user distributions without sleeping capabilities. To the best of the authors' knowledge, no work in SG studies the benefit of SS using this more realistic network topology. 

In addition to inhomogeneous deployments, future mobile networks are expected to include massive multiple-input multiple-output (MIMO) and millimeter wave (mm-Wave) transmissions in order to meet the requirements beyond 5G \cite{7397887}. The authors of \cite{7913628} provide a framework to study the impact of beamforming and mm-Wave on the coverage probability without sleep modes. In \cite{7792582}, the performance of HetNets where SS and mm-Wave feature is analyzed.


\begin{table*}[t]
\renewcommand{\arraystretch}{1.3}
\caption{\textcolor{\CommentsColor}{Comparison of relevant literature with this paper}}
\label{tab:SotA}
\centering
\begin{tabular}{|c|c|c|c|c|c|c|}
\hline
\rowcolor[HTML]{C0C0C0} 
\textbf{Ref.}  & \textbf{\begin{tabular}[c]{@{}c@{}}Random \\ Sleeping\end{tabular}} & \textbf{\begin{tabular}[c]{@{}c@{}}Strategic \\ Sleeping\end{tabular}} & \textbf{HetNet} & \textbf{\begin{tabular}[c]{@{}c@{}}Non-Uniform\\ UEs\end{tabular}} &  \textbf{mm-Wave} & \textbf{Beamforming}\\ \hline
\cellcolor[HTML]{FFFFFF} \cite{6566900,7036048}  & \checkmark &  &  &   &  & \\ \hline
\rowcolor[HTML]{EFEFEF} \cite{6575091,6774407,7136461} & \checkmark  &  & \checkmark  &  &   & \\ \hline
\cellcolor[HTML]{FFFFFF} \cite{7185349} & \checkmark   &  & \checkmark   &  & & \checkmark \\ \hline
\rowcolor[HTML]{EFEFEF} \cite{7248320,8458429} &   & \checkmark  &     &   &  &  \\ \hline
\cellcolor[HTML]{FFFFFF} \cite{6567876,7057551,7881136,7475871,7060678,8372969,8633858,6502479,9676414} &   &  \checkmark  & \checkmark &    &   &    \\ \hline
\rowcolor[HTML]{EFEFEF} \cite{6684701}  &   &  \checkmark & \checkmark  & \checkmark  &  &  \\ \hline
\rowcolor[HTML]{FFFFFF} \cite{7792582}  &   &  \checkmark   & \checkmark  &   & \checkmark &  \\ \hline
\rowcolor[HTML]{EFEFEF} This work &   & \checkmark   & \checkmark  & \checkmark  & \checkmark   & \checkmark \\ \hline
\end{tabular}
\end{table*}

\subsection{Major Contributions}
On the basis of the aforementioned works, we provide a mathematical framework to study sleep strategies using SG. This framework includes advanced network features, which are summarized in the last line of Table \ref{tab:SotA}. The contributions associated to these features can be detailed as follows:
\begin{enumerate}
    \item Our system model extends the networks considered in prior SG works studying sleep control. This generalization is built up on three major axes.
    \begin{itemize}
            \item The topology of the network nodes is improved to encompass more realistic scenarios. In addition to modeling base stations using a multi-tier architecture, users locations are sampled from non-uniform distributions. The resulting framework enables to model hotspots (PCPs of users around BSs), areas with a temporarily high activity at a given instant (independent PCPs of users), as well as areas with a low traffic activity (independent Homogeneous Poisson Point Processes (HPPPs)).
            \item mm-Wave transmissions and directional beamforming are included in the propagation model. Nakagami fading is considered since it captures poor scattering environments with better accuracy than Rayleigh distributions \cite{7593259}. Regarding the beamforming, this study capitalizes on the efficient approximation used in \cite{7913628} to model the antenna pattern. This approximation has been shown to provide better trade-offs in terms of tractability and accuracy compared to the flat-top model employed in most existing works \cite{7308984,7448962,6932503,7105406}.
            \item \textcolor{\CommentsColor}{A network aware sleep strategy is studied. This strategy is based on the user activity within each Voronoï cell. A probabilistic model is defined as function of a general random variable representing this activity and determines the mode, asleep or active, selected by each base station.} 
\end{itemize}       
 To the best of the authors' knowledge, this work is the first SG analysis dedicated to sleep control considering these aspects.
    \item A general and tractable framework including the above features is presented. This mathematical framework also introduces new metrics related to the power strategies. 
        \begin{itemize}
            \item A novel expression is provided for the distribution of the number of users in Voronoï cells. This distribution captures the simultaneous presence of users associated to both HPPPs and PCPs. At the moment of writing this study, such a two-fold load characterization has never been considered. This load is then used to represent the cellular activity and employed as a criterion to select the sleep mode. 
            \item We introduce the Activity Averaged K-tier Coverage Probability (AAKCP) to characterize the Quality-of-Service (QoS). This probability characterizes the global network coverage, averaged over the different sleep modes and user categories. On the basis of this metric, revisited definitions of the spectral and energy efficiencies are also provided.
            \item \textcolor{\CommentsColor}{In addition to the complete expressions, some asymptotic expressions have been derived to drive more insights into the impact of the system parameters.}
            \item The developed framework is validated by means of extensive Monte Carlo (MC) simulations. The numerical results enable to gain insight into the impact on the performance of the activity thresholds defining the different power levels.
        \end{itemize}
\end{enumerate}

\textcolor{\CommentsColor}{
Please note that unlike Table I may suggests, this study does not only merge the findings from the relevant literature. Instead, the full mathematical expression of the AKKCP metric is developed from scratch. Thus, it is not just an extension of the results presented in the cited papers.}
\subsection{Mathematical Background}
\label{sec:math_back}
In SG studies, nodes distributions with independent locations are traditionally analyzed by means of HPPPs. Beyond these baseline processes, advanced patterns featuring attraction behaviors between nodes can also be modeled using cluster processes, such as the PCP. For $\lambda_p, \overline{m} \in \mathbb{R}$ and function $f$, a PCP $\Phi$ can be defined as follows. First, a parent Point Process (PP) $\Phi_p$ is modeled by mean of an HPPP of intensity $\lambda_p$. Each parent point $\mathbf{z} \in \Phi_p$ form a cluster center and is associated with a daughter PP $\Phi_{d,\mathbf{z}}$ \cite{8187697}. For each daughter PP, the random counting measure\footnote{The random counting measure is defined as a function $\Psi(\cdot)$ that takes a subset $A \subset \mathbb{R}^2$ as input and returns the number of points of $\Phi$ within this area. By denoting the indicator function as $\mathbbm{1}(\cdot)$, the random counting measure can be mathematically defined as follows:
\begin{equation}
\label{eq:chap1_rcm}
        \Psi : \mathbb{R}^2  \to \mathbb{N} :  A \to \sum_{\mathbf{x} \in \Phi} \mathbbm{1}(\mathbf{x} \in A).
\end{equation}} $\Psi$ is assumed to be given by an independent Poisson random variable of parameter $\overline{m}$.
The positions of the daughter points with respect to their cluster center $\mathbf{z}$ are given by independent and identically distributed random vectors $\mathbf{s} \in \Phi_{d,\mathbf{z}}$. The PCP is given by $\Phi = \bigcup_{\mathbf{z} \in \Phi_p} \mathbf{z} + \Phi_{d,\mathbf{z}}$. One will denote, by $f(\mathbf{s})$, the probability density function (PDF) of the vectors $\mathbf{s}$. Two different categories of PCPs can be defined depending on this distribution, namely Thomas Cluster Processes (TCPs) and Matérn Cluster Processes (MCP). \color{\CommentsColor}In this paper, we only consider MCPs but the framework can be directly extended to TCPs. A PCP $\Phi(\lambda_p,\overline{m},f)$ in $\mathbb{R}^2$ is an MCP if its daughter points $\mathbf{s} \in \Phi_{d,\mathbf{z}}$ are uniformly distributed within a disk of radius $r_{m}$ centered at $\mathbf{0}$, which is denoted as $b(\mathbf{0},r_m)$. $f$ can in that case be expressed as $f(\mathbf{s}) = \frac{1}{\pi r_{m}^2}, 0 \leq \lVert \mathbf{s} \rVert \leq r_{m}.$\color{black}

\subsection{Structure of the Paper}
The paper is structured as follows. The system model is presented in Section \ref{sec:syst_mod}. The analytical results using SG are presented in Section \ref{sec:analy_res}. These analytical findings are validated in Section \ref{sec:simu} by means of MC simulations. The impact of the sleep control on the EE is studied in the same section. Finally, Section \ref{sec:conclu} concludes this work.
\section{System Model}
\label{sec:syst_mod}
This section is divided in three parts. First, the network topology is described. Then, the policy for the sleep control is presented. Finally, the channel model is defined.
\subsection{Network Model}
\label{sec:net_mod}
Base stations are here modeled using a HetNet consisting in $K^\text{BS}$ tiers. Let $\mathcal{K}^\text{BS}$ be the set of indices of these tiers. The subset of tiers modelling the base network independent of the UEs locations is denoted as $\mathcal{K}_\textup{B}^\text{BS} \subseteq \mathcal{K}^\text{BS}$. While $\mathcal{K}_\textup{H}^\text{BS} \subseteq \mathcal{K}^\text{BS}$ denotes the subset of hotspots dependent on UEs locations. Each tier $k\in \mathcal{K}^\text{BS}$ is defined as an HPPP $\Phi_k^\text{BS}$ of density $\lambda_k^\text{BS} > 0$. The BSs of this tier $k$ transmit with maximum power $P_k > 0$ and are equipped with a Uniform Linear Array (ULA) of $M_k$ antennas. Universal frequency reuse is considered such that the BSs interfere with each other. A nearest neighbor open access association scheme is used such that each UE is connected to its nearest non-sleeping BS. This notion of sleep mode is discussed in Subsection \ref{sec:PowerControl}.

To model a non-uniform distribution, $K^\text{UE}$ tiers of UEs are considered. 
Let $\mathcal{K}^\text{UE}$ be the set of indices of these tiers. Each tier $u\in\mathcal{K}^\text{UE}$ belongs to one of the three categories considered in this paper:
\begin{itemize}
    \item \textbf{Category 1:} The objective of this category is to model locations with a temporary high user activity, independent of the BS locations. The UEs are distributed as \textcolor{\CommentsColor}{an MCP} $\Phi_u^\text{UE}$ with a parent PP independent of the BS locations. The intensity of the parent PP is $\lambda_{p,u}^\text{UE}$. The mean number of UEs per cluster is given by $\overline{m}_u$. The density of UEs is $d_u = \overline{m}_u \lambda_{p,u}^\text{UE}$ and the corresponding \textcolor{\CommentsColor}{radius is denoted $r_{m,u}$.}
    \item \textbf{Category 2:} The role of this category is to model hotspots: users gathered around fixed BSs. For each tier $k \in \mathcal{K}_\text{H}^\text{BS}$, \textcolor{\CommentsColor}{an MCP} $\Phi_{u,k}^\text{UE}$ of parent PP $\Phi_k^\text{BS}$ is defined. The intensity of the parent PP is thus $\lambda_{p,u}^\text{UE}=\lambda_k^\text{BS}$. The other parameters are the same as for category 1 and can take any value.
    \item \textbf{Category 3:} This category models uniformly distributed UEs. The UEs are distributed as an independent HPPP $\Phi_u^\text{UE}$ of density $d_u = \lambda_u^\text{UE} > 0$. 
\end{itemize} 
\begin{figure*}[!t]
\centering
\subfloat[]{\includegraphics[width=0.48\textwidth]{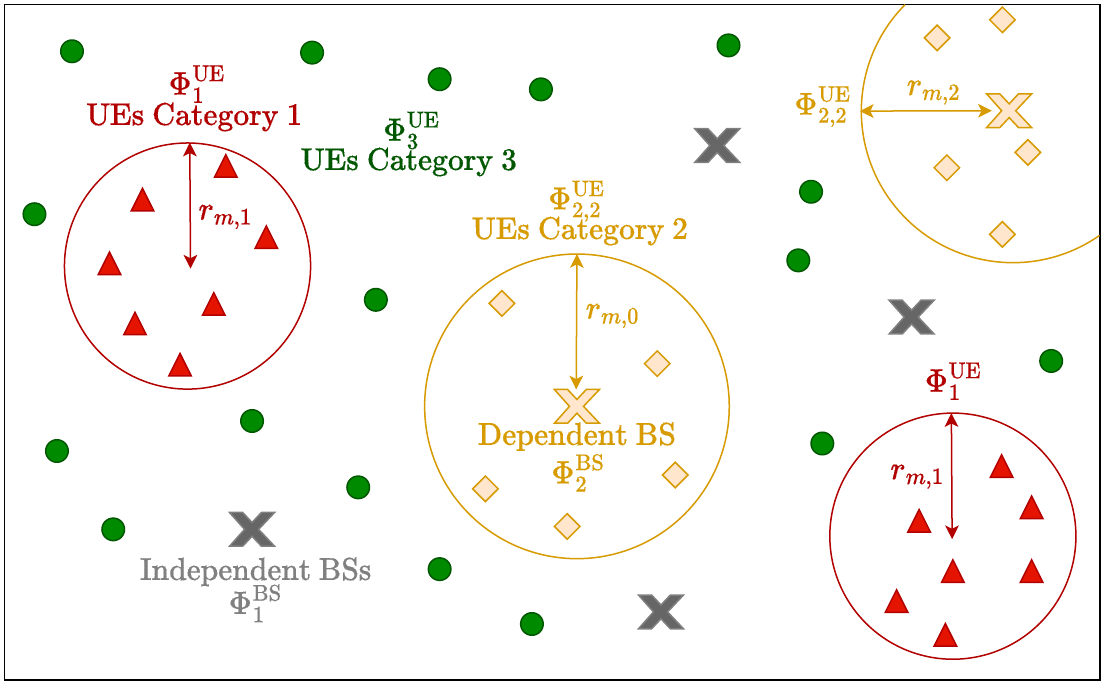}}
\label{fig:first_case}
\hfill
\subfloat[]{\includegraphics[width=0.48\textwidth]{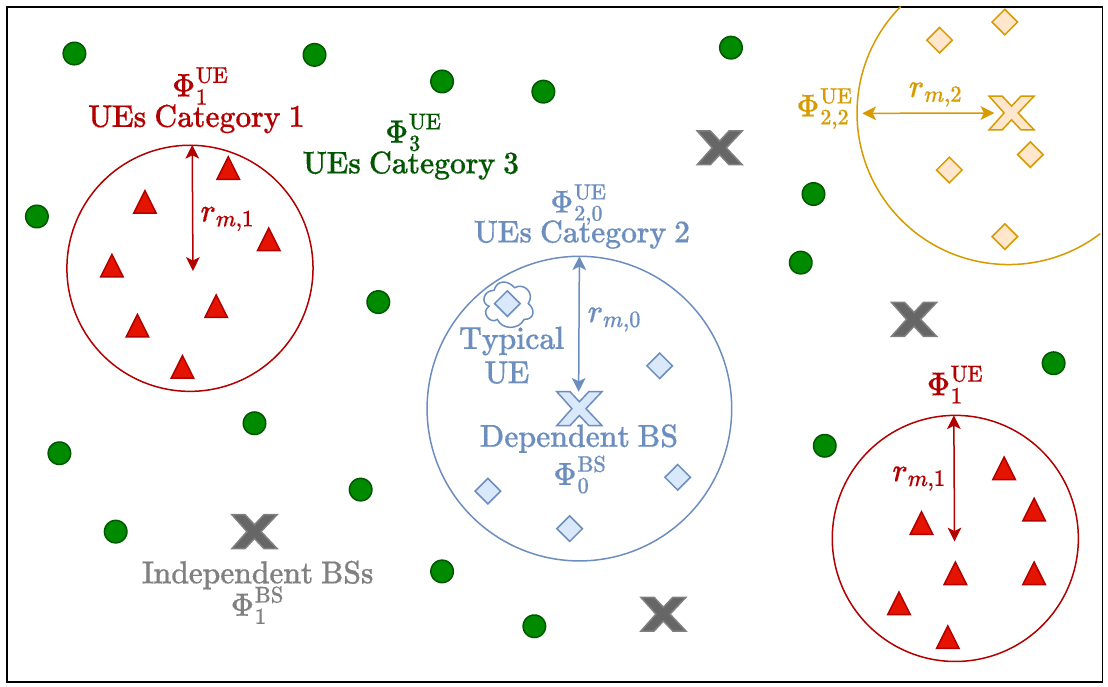}}
\label{fig:second_case}
\caption{\textcolor{\CommentsColor}{Sample of the network's topology. (a) network realization with the three categories of UEs' tiers and the two types of BSs tiers are represented. (b) illustration of the tier $0$ when the typical UE is from a tier of category 2.}}
\label{fig:topo}
\end{figure*}
The set of tiers belonging to each category is denoted by $\mathcal{K}_\textup{U}^\textup{UE}$ with $\textup{U} \in \{1,2,3\}$. \figurename \ref{fig:topo} illustrates one possible realization of the network. The analysis of the performance of the network is done for a randomly selected UE (called the typical UE) among one of the $K^\text{UE}$ tiers of UEs. Without loss of generality, the typical UE is located at the origin. Let us define $i \in \hat{I}_k$ as the index of the $i^{th}$ closest BS to the typical UE and $\hat{I}_k$ the set of those indices. When the coverage is studied for a typical UE from a tier $u$ of category 2, one BS is present at the center of its representative cluster. Let $\mathbf{y} \in \Phi_k^\text{BS}$ be the position of this BS. To simplify the notation, another tier $\Phi_0^\text{BS}$, which consists only of this single point $\mathbf{y}$, is defined as in \cite{7511509}. Using Slivnyak's Theorem, the thinned process $\Phi_k^\text{BS} \setminus \{\mathbf{y}\}$ has the same distribution as $\Phi_k^\text{BS}$. The tier referred to by index $k$ will thus be the thinned one $\Phi_k^\text{BS} \setminus \{\mathbf{y}\}$. The parameters of the $0^{th}$ tier are derived from $\Phi_{u,k}^\text{UE}$ and $\Phi_k^\text{BS}$ such that $P_0=P_k$, $M_{0} = M_{k}$ and \textcolor{\CommentsColor}{$r_{m,0} = r_{m,u}$}. A new set of indices is defined as $\mathcal{K}_0^\text{BS} = \{0\} \bigcup \mathcal{K^\text{BS}}$. For the sake of coherence, when the coverage is studied for a typical UE from a tier of category 1 or 3, $\Phi_0^\text{BS} = \emptyset$. 

\subsection{Power Control Policy}
\label{sec:PowerControl}
Base stations are able to \textcolor{\CommentsColor} {turn on and off} according to the activity in their cells. The activity is thus a measure of the workload of a BS at a given time. Let $\mathcal{A}_{k,i}$ be the activity of the $i^{th}$ closest BS of tier $k$ to the typical user. In this work, this quantity is defined as a discrete random variable in the set $\mathbb{A}$, representing the number of UEs within the Voronoï region of a BS. Note that all results can be generalised to another definition of the activity (e.g. defined as a continuous variable).

For each BS tier $k \in \mathcal{K}^\text{BS}$,\color{\CommentsColor} the power control level is defined as a step function $s_k:\mathbb{A} \to \textcolor{\CommentsColor}{\{0,1\}}: a \to s_k(a)$ defined as:
\begin{equation}
    s_k(a)=\begin{cases}
  \multicolumn{1}{@{}c@{\quad}}{0} & \text{ for } a < \mu_{k}, \\ 
  \multicolumn{1}{@{}c@{\quad}}{1} & \text{ for } a \geq \mu_{k}.
\end{cases}
\end{equation}

Ratios are defined because the thresholds $\mu_{k}$ lose their interpretation when one changes the density of BS and/or UEs. Therefore, the corresponding ratio of BSs in active mode is denoted as $q_{k}$. When the same ratio $q_{k}$ is used for every tier of the network, it is denoted as $q$. However, it does not lead to the same thresholds $\mu_{k}$ as all tiers of BSs might have different PMFs characterizing their load.
\color{black}

In order to study the benefit of sleep control, a power consumption model must be defined. For a BS of tier $k$ in sleep mode, the consumed power is defined as $P_{\text{sleep},k}$. For active BSs, one can define the static, antenna dependent and load dependent consumed powers, respectively denoted as \textcolor{\CommentsColor}{$P_{\text{stat},k}$, $M_{k}P_a$ and $\Delta_p P_{k}$} \cite{9013130}. In the last term, $\Delta_p$ is the slope of the load dependent power consumption. For BSs with sleeping capabilities $P_{\text{stat},k} > P_{\text{sleep},k}$. The average power consumption $P_{c,k}$ per BS of tier $k$ is given by 
\color{\CommentsColor}
\begin{equation} \label{eq:Power_consumed}
    P_{c,k} = (1-q_{k}) P_{\text{sleep},k} + q_{k}(P_{\text{stat},k} + M_{k} P_a + \Delta_p P_{k}) ~[\text{W}]
\end{equation}
\color{black}
The average network power consumption per surface unit $P_\text{net}$ is thus obtained as follows \cite{7377022}: 
\begin{equation}
    P_\text{net} = \sum_{k \in \mathcal{K}^\text{BS}} \lambda_k^\text{BS} P_{c,k} ~\left[\frac{\text{W}}{\text{m}^2}\right].
\end{equation}

\subsection{Channel Model}
\begin{figure}
\centering
\includegraphics[width=0.49 \textwidth]{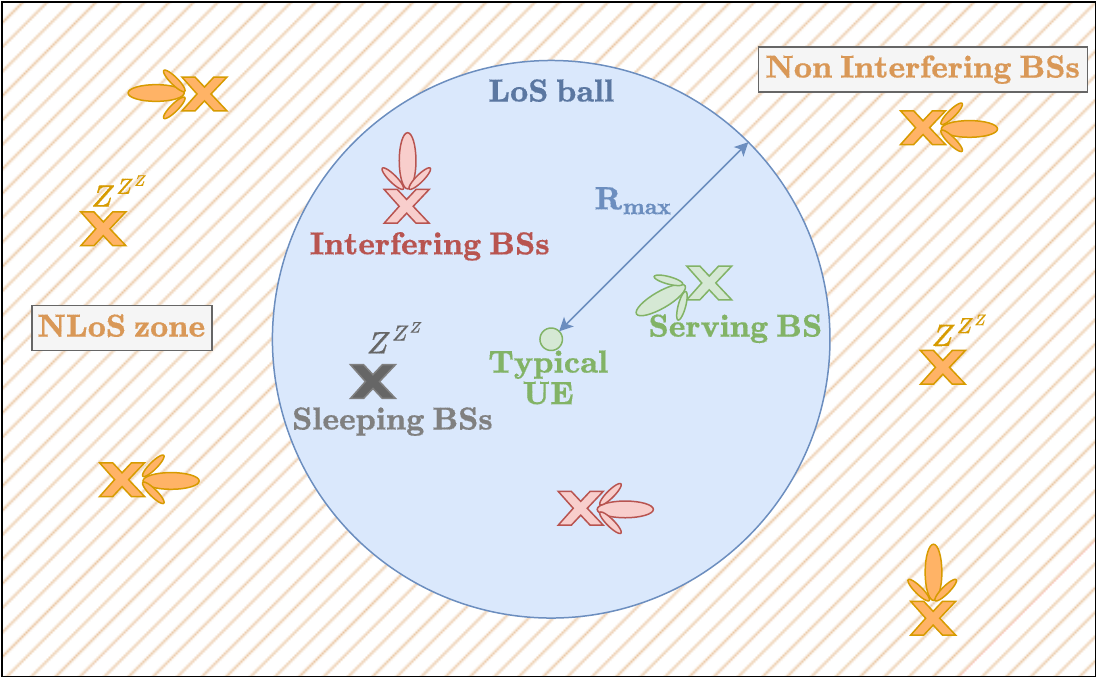}
\caption{\textcolor{\CommentsColor}{Representation of the channel model.}}
\label{fig:chan_mod}
\end{figure}
A mm-Wave propagation model is considered. In this work, only Line-of-Sight (LoS) transmissions are encountered, since non-LoS contributions have been shown to be negligible for mm-Waves \cite{7448962}. To focus the analysis on LoS transmissions to a typical UE located at the origin, the LoS ball blockage model is used \cite{6932503}. The transmissions from BSs located within a ball of radius $R_\text{max}$ are considered as LoS, while the ones from BSs outside this ball are fully neglected. \textcolor{\CommentsColor}{\figurename \ref{fig:chan_mod} depicts these considerations.} For a serving BS of tier $j \in \mathcal{K}_0^\text{BS}$, located at $\mathbf{x}_j$, the received signal at the typical UE is given by
\color{\CommentsColor}
\begin{align*}
y = & \sqrt{\beta \textcolor{\CommentsColor}{P_{j}}} \lVert \mathbf{x}_j \rVert^{-\frac{\alpha}{2}} \mathbf{h}_{\mathbf{x}_j}^\textup{H} \mathbf{w}_{\mathbf{x}_j} s_{\mathbf{x}_j} + n_0\\
    & + \sum_{k\in\mathcal{K}_0} \sum_{\substack{\mathbf{x}_k \in \\ \Tilde{\Phi}_{k}^\text{BS}(\mathbf{x}_j)}} \sqrt{\beta \textcolor{\CommentsColor}{P_{k}}} \lVert\mathbf{x}_k \rVert^{-\frac{\alpha}{2}} \mathbf{h}_{\mathbf{x}_k}^\textup{H} \mathbf{w}_{\mathbf{x}_k } s_{\mathbf{x}_k}, \numberthis
\end{align*}
where $\mathbf{h}_{\mathbf{x}_k}, \mathbf{w}_{\mathbf{x}_k} \in \mathcal{C}^{M_k \times 1}$ \color{black} are respectively the channel and beamforming vectors between a transmitting BS of tier $k\in \mathcal{K}_0^\textup{BS}$ located at $\mathbf{x}_k$ and the typical UE. The Additive White Gaussian Noise (AWGN) of variance $\sigma_\textup{noise}^2$ is represented by $n_0$. The transmitted signal is denoted as $s_{\mathbf{x}_k}$. The path loss intercept and exponent are respectively $\beta$ and $\alpha$ and a circular excluding region of radius $R_\text{min}$ is defined to circumvent singularities. Potential BSs in these regions are hence ignored. \textcolor{\CommentsColor}{The set of interfering BSs of tier $k$ is $\Tilde{\Phi}_{k}^\text{BS}(\mathbf{x}_j)$ obtained after a thinning operation of probability $q_k$} . As in \cite{1146527}, the effect of the channel is characterized by \textcolor{\CommentsColor}{$\mathbf{h}_{\mathbf{x}_k} = \sqrt{M_k} \rho_{\mathbf{x}_k} \mathbf{a}_t(\psi_{\mathbf{x}_k})$}, where $\rho_{\mathbf{x}_k}$ is the complex small-scale fading gain and $\mathbf{a}_t(\psi_{\mathbf{x}_k})$ is the transmit array vector that can be written as
\begin{equation}
\label{eq:chap4_at}
        \mathbf{a}_t(\psi_{\mathbf{x}_k}) = \frac{1}{\sqrt{M_k}} \left[1, \cdots, e^{j2\pi l \psi_{\mathbf{x}_k}}, \cdots, e^{j2\pi (M_k-1) \psi_{\mathbf{x}_k}}  \right]^T,
\end{equation}
where $\psi_{\mathbf{x}_k} = \frac{d}{\upsilon} \cos \phi_{\mathbf{x}_k}$, the angle of departure from the BS to the served UE is $\phi_{\mathbf{x}_k}$, the distance between antennas in the ULA is $d$ and $\upsilon$ is the wavelength. The beamforming vector is given by $\mathbf{w}_{\mathbf{x}_k} = \mathbf{a}_t(\varphi_{\mathbf{x}_k})$, where $\varphi_{\mathbf{x}_k}$ is a variable depending on the direction of the beam. In this paper, Nakagami-$m$ fading is considered for $\lvert \rho_{\mathbf{x}_k} \rvert$, as in \cite{6932503}. The power gain corresponding to the channel effect and the beamforming is given by:
\begin{equation}
\label{eq:gain}
    |\mathbf{h}_{\mathbf{x}_k} \mathbf{w}_{\mathbf{x}_k} |^2 = M_k |\rho_{\mathbf{x}_k}|^2 |\mathbf{a}_t^{H}(\psi_{\mathbf{x}_k})  \mathbf{a}_t(\varphi_{\mathbf{x}_k})|^2,
\end{equation}
where $|\rho_{\mathbf{x}_k}|^2$ is the power gain of the small-scale fading. In the case of Nakagami-$m$ fading, it follows a gamma distribution denoted as Gamma($m,\frac{1}{m}$). To maximize the received power, the BS serving the typical UE aligns its beam such that $\varphi_{\mathbf{x}_k}=\psi_{\mathbf{x}_k}$. For the interfering BSs, we assume that $\psi_{\mathbf{x}_k},\varphi_{\mathbf{x}_k}$  are uniformly distributed over $[-\frac{d}{\upsilon},\frac{d}{\upsilon}]$. The array gain $|\mathbf{a}_t^{H}(\psi_{\mathbf{x}_k})  \mathbf{a}_t(\varphi_{\mathbf{x}_k})|^2$ in \eqref{eq:gain} can be expressed as a kernel $G_{act}(\psi_{\mathbf{x}_k}-\varphi_{\mathbf{x}_k})$ with $G_\text{act}(x) \triangleq \frac{\sin^2(\pi M_k x)}{M_k^2 \sin^2(\pi x)}$. In \cite[Appendix A]{7279196}, it has been shown that for $\frac{d}{\upsilon}=\frac{1}{2}$, the term $\psi_{\mathbf{x}_k}-\varphi_{\mathbf{x}_k}$ in the expression of the kernel can equivalently be replaced by $\frac{d}{\upsilon} \theta_{\mathbf{x}_k}$, where $\theta_{\mathbf{x}_k} \sim \text{Unif}[-1,1]$. This leads to the first assumption of this work.
\begin{assumption}
\label{assum: d/lambda}
The BSs are equipped with half-wavelength antennas in ULA form.
\end{assumption}

The actual kernel has to be approximated as it can not always lead to a closed form expression with SG analysis. In \cite{7913628}, a truncated cosine kernel is shown to provide a trade-off between tractability and accuracy. For this reason, this kernel is used as an approximation in our framework. 
\begin{assumption}
The actual beamforming kernel $G_\text{act}(x)$ is approximated by a cosine kernel defined as
\begin{equation}
    \label{eq:chap4_cospat}
        G_\text{cos}(x) = \begin{cases}
\multicolumn{1}{@{}c@{\quad}}{\cos^2\left(\frac{\pi M_k}{2}x \right)}           & \text{ if } |x| \leq \frac{1}{M_k}, \\ 
          \multicolumn{1}{@{}c@{\quad}}{0} & \text{ otherwise}.
        \end{cases}
\end{equation}
\end{assumption}
The channel gain term including both the small-scale fading gain and the directional antenna array gain is then defined as $g_{\mathbf{x}_k}=M_k|\rho_{\mathbf{x}_k}|^2 G_\text{cos}(\frac{d}{\lambda} \theta_{\mathbf{x}_k})$.

\section{Analytical Results}
\label{sec:analy_res}
\begin{figure*}[!t]
    \centering
    \includegraphics[width=0.7\textwidth]{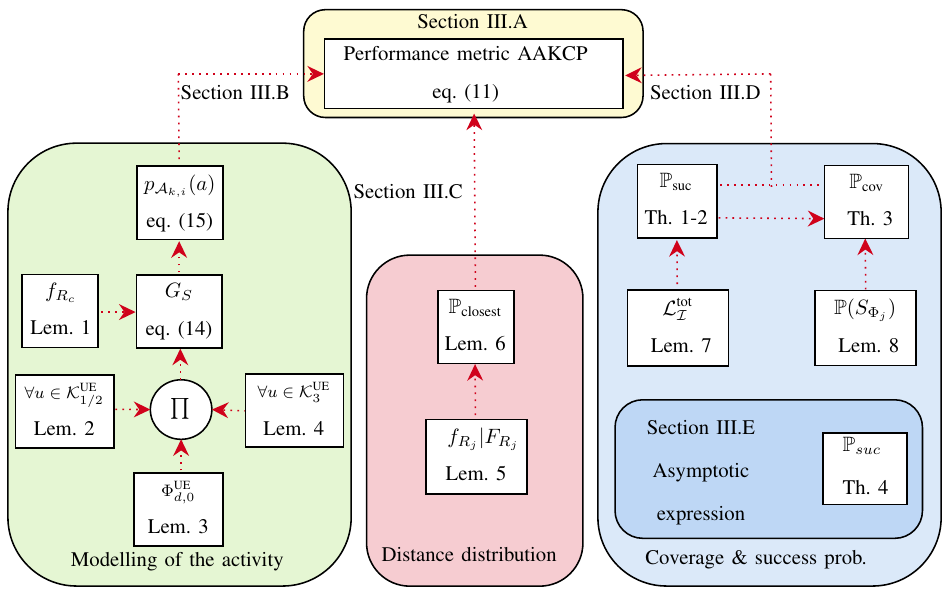}
    \caption{\textcolor{\CommentsColor}{Summary of the analytical findings.}}
    \label{fig:diagram_AAKCP}
\end{figure*}
The structure of this section is summarized in \figurename \ref{fig:diagram_AAKCP}. In Subsection \ref{sec:metric}, a revisited definition of the coverage probability is introduced on the basis of the system model. This new performance metric is the Activity Averaged K-tier Coverage Probability. To evaluate this metric, the analysis is divided into three parts: 
\begin{itemize}
    \item Subsection \ref{sec:activity} contains the statistical characterization of the cellular load. To model the PMF of the activity $p_{\mathcal{A}_{k,i}}(a)$, the Probability Generating Function (PGF) of the load $G_S$ must be evaluated. This analytical development relies on the assumption of circular Voronoï region detailed in the subsection. Its radius PDF is denoted as $f_{R_c}$.
    \item Subsection \ref{sec:distances} provides an expression for the distribution of the distance to the typical UE. \textcolor{\CommentsColor}{The PDF $f_{R_{j}}$ and the Complementary Cumulative Distribution Function (CCDF) $F_{R_{j}}$ of the distance to the closest closest BS of tier $j$ are derived.} These quantities are used to evaluate the probability $\mathbb{P}_\text{closest}$ that a BS of a given tier is the nearest to the typical UE. 
    \item Subsection \ref{sec:perf} aims at calculating tractable expressions of the coverage and success probabilities. To do so, the expression of the Laplace Transform (LT) of the interference $\mathcal{L}_\mathcal{I}^\text{tot}$ is derived. Finally, the probability to be connected to the closest awake BS from tier $j\in \mathcal{K}_0^\text{BS}$ is given by $\mathbb{P}(S_{\Phi_j})$.
\end{itemize}

\textcolor{\CommentsColor}{Please note that in order to derive closed-form analytical expressions, several assumptions are stated in the development (i.e. Assumption 1 to 6). Their impact on the accuracy of the obtained expression is studied by means of MC simulations in section \ref{sec:simu}.}

\subsection{Performance Metrics}
\label{sec:metric}
To study the average performance of homogeneous networks, the success and coverage probabilities are commonly employed \cite{9378781}. The success probability $\mathbb{P}_\text{suc}(\tau,\lVert \mathbf{x} \rVert)=\mathbb{P}(\text{SINR}(\lVert \mathbf{x} \rVert) > \tau)$ is here defined as the probability for the SINR to be greater than a QoS threshold $\tau$ when the serving BS is at a given distance $\lVert \mathbf{x} \rVert$. The coverage probability is obtained by taking the average of this success probability over the distance to the serving BS and is given by $\mathbb{P}_\text{cov}(\tau) = \mathbb{E}_{\lVert \mathbf{x} \rVert} \left[\mathbb{P}_\text{suc}(\tau,\lVert \mathbf{x} \rVert) \right]$.

For heterogeneous networks, the typical user can be served by the closest BS of each tier. Let us denote the tier of the nearest BS as $k^{*} \in \mathcal{K}_0^\text{BS}$. The heterogeneous coverage probability follows as
\begin{equation}
\label{eq:cov_het}
     \mathbb{P}_\text{cov}^\textup{het}(\tau)
     =  \sum_{k^{*} \in \mathcal{K}_0^\text{BS}} \mathbb{E}_{\textcolor{\CommentsColor}{R_{k^{*}}}} \Big[ \mathbb{P}_\text{closest}(k^{*}|\textcolor{\CommentsColor}{R_{k^{*}}}) \mathbb{P}_\text{suc}^{k^{*}}(\tau,\textcolor{\CommentsColor}{R_{k^{*}}}) \Big],
\end{equation}
with $\textcolor{\CommentsColor}{R_{k^{*}}}$ the distance to the nearest BS of tier $k^*$ and $\mathbb{P}_\text{closest}(k^{*})$ defined as the probability that the nearest BS is from tier $k^{*}$.
\color{\CommentsColor}
In the framework of this paper, sleep control is applied as described in Subsection \ref{sec:PowerControl}. The consequence is that the closest BS is not always the one serving as it can be sleeping when $s_k(\mathcal{A}_{k^{*},1})=0$. A network success probability is thus defined by considering an active case in which the typical UE is served by the nearest BS of tier $k^{*}$ at distance $R_{k^{*},1}$, and a sleeping case in which it is served by BSs further away is then added. Knowing the activity level $\mathcal{A}_{k^{*},1}$, the network success probability is thus defined as
\begin{align*}
    \mathbb{P}_\text{suc}^\textup{net}(\tau,& R_{k^{*}}|k^{*},\mathcal{A}_{k^{*},1}) =  \mathbbm{1}_{(\mathcal{A}_{k^{*},1} \geq \mu_{k^*})}\mathbb{P}_\text{suc}^{k^{*}}(\tau,R_{k^{*}})\\
    & + \mathbbm{1}_{(\mathcal{A}_{k^{*},1} < \mu_{k^*})}\mathbb{P}_\text{cov}(\tau|S_{\Phi_{k^{*}}}',R_{k^{*}}) \numberthis \label{eq:suc_net},
\end{align*}
\color{black}
\textcolor{\CommentsColor}{in which $S_{\Phi_{k^{*}}}$ (resp. $S_{\Phi_{k^{*}}}'$)  denotes the event that the typical UE is  connected (resp. not connected) to the first BS of tier $k^{*}$.} Regarding the last term of \eqref{eq:suc_net}, the success probability must be averaged over their distance to the typical user since the locations of the BSs further than the nearest one are random. This leads to the coverage probability term \textcolor{\CommentsColor}{$\mathbb{P}_\text{cov}(\tau|S_{\Phi_{k^{*}}'},R_{k^{*}})$} when the typical UE is not connected to the closest BS.

\begin{figure*}[!t]
\normalsize
\setcounter{MYtempeqncnt}{\value{equation}}
\setcounter{equation}{10} 
\color{\CommentsColor}
\begin{align*}
\mathbb{P}_\textup{AAKCP}(\tau|u) & = \sum_{k^{*} \in \mathcal{K}_0^\text{BS}} \mathbb{E}_{\textcolor{\CommentsColor}{R_{k^{*}}}}  \bigg[\mathbb{P}_\text{closest}(k^{*} |\textcolor{\CommentsColor}{R_{k^{*}}}) \mathbb{E}_{\mathcal{A}_{k^{*},1}} \Big[ \mathbb{P}_\text{suc}^{\textup{net}}(\tau,\textcolor{\CommentsColor}{R_{k^{*}}}|k^{*},\mathcal{A}_{k^{*},1})  \Big]  \bigg] \\
    = & \sum_{k^{*} \in \mathcal{K}_0^\text{BS}} \int_{R_\text{min}}^{R_\text{max}}   \mathbb{P}_\text{closest}(k^{*} | \textcolor{\CommentsColor}{r_{k^{*}}}) \bigg[
    \mathbb{P}_{k^*,1}(1|\textcolor{\CommentsColor}{r_{k^{*}}})
     \mathbb{P}_\text{suc}^{k^{*},n} \left(\tau,\textcolor{\CommentsColor}{r_{k^{*}}}\right) \\
     & + \mathbb{P}_{k^*,1}(0|\textcolor{\CommentsColor}{r_{k^{*}}})
    \mathbb{P}_\text{cov}\left(\tau \Big|\textcolor{\CommentsColor}{S_{\Phi_{k^{*}}}'}, \textcolor{\CommentsColor}{r_{k^{*}}} \right) 
    \bigg] f_{\textcolor{\CommentsColor}{R_{k^{*}}}}(\textcolor{\CommentsColor}{r_{k^{*}}}|R_\text{min})  \textup{d}\textcolor{\CommentsColor}{r_{k^{*}}},\numberthis \label{eq:AAKCP_u}
\end{align*}  
\color{black}
\setcounter{equation}{\value{MYtempeqncnt}}
\hrulefill
\vspace*{4pt}
\end{figure*}
\addtocounter{equation}{1}
Replacing the success probability in \eqref{eq:cov_het} by \eqref{eq:suc_net} and averaging with respect to the activity level yields the Activity Averaged K-tier Coverage Probability (AAKCP). For a typical user of a given $u \in \mathcal{K}^\text{UE}$, this metric is defined by \eqref{eq:AAKCP_u} \textcolor{\CommentsColor}{where $\mathbb{P}_{k^*,1}(0|r_{k^{*}})$ and $\mathbb{P}_{k^*,1}(1|r_{k^{*}})$ represent the probability of the nearest BS being asleep or awake, respectively, based on the distance $r_{k^{*}}$ to that BS.}

The tier $u$ to which the typical user belongs has an impact on the success probability and the distributions of the distances and the activities. The AAKCP for a mean user from any tier can then be obtained as
\begin{equation}
\label{eq:AAKCP_gen}
    \mathbb{P}_\textup{AAKCP}(\tau) = \frac{1}{\lambda_\text{tot}^\text{UE}}\sum_{u \in \mathcal{K}^\text{UE}} d_u \mathbb{P}_\textup{AAKCP}(\tau|u),
\end{equation}
with $\lambda_\text{tot}^\text{UE}=\sum_{u \in \mathcal{K}^\text{UE}} d_u$. The Area Spectral Efficiency (ASE) which represents the network throughput is defined as $\mathcal{T} = \lambda_\text{tot}^\text{BS}\mathbb{P}_\textup{AAKCP}(\tau) \log_2(1+\tau)$ expressed in (bit/s)/Hz$\cdot \text{m}^2$, where $\lambda_\text{tot}^\text{BS}=\sum_{k\in \mathcal{K}^\text{BS}} \lambda_k^\text{BS}$ is the network density of BS. Lastly, the Energy Efficiency (EE) is defined as:
\begin{equation} \label{eq:EE}
    \text{EE} = \frac{\mathcal{T}}{P_\text{net}} ~\left[\frac{\text{bit}}{\text{Hz J}} \right].
\end{equation}
\subsection{Modeling of the Cellular Activity}
\label{sec:activity}
The activity of a BS is here defined as the number of UEs within its Voronoï region before any sleeping strategy. In our framework, the activity is thus a discrete variable. For homogeneous networks, the distribution of the number of UEs per cell is studied in \cite{7433666}, \cite{9079449} and \cite{9443319,9887891} for homogeneous, clustered and spatially coupled users respectively. However, no characterization has been performed for HetNets with non-uniform user distribution. \textcolor{\CommentsColor}{Note that this definition of the activity does not consider coordination between BSs that enables offloading methods. The objective is to present a straightforward SG framework demonstrating the advantages of activity-based power control. Other definitions of the activity, such as the one presented in \cite{6684701}, could be used to evaluate the AAKCP metric, provided that a pmf of this activity can be derived.}

Other definitions of the activity can be used to derive its pmf. For instance, one could adapt the expression obtained through Markov Theory in \cite{6684701} based on the considered network topology.

In order to evaluate the AAKCP metric, the PMF $p_{\mathcal{A}_{k,i}}(a)$ of the random activity of every BS $i$ of each tier $k \in \mathcal{K}_0^\text{BS}$ has to be evaluated. In this subsection, unlike the others of this paper, the characterization is derived for a typical BS located at the origin and not for a typical UE. Indeed, this subsection aims at deriving properties of a BS. In the second part of this subsection, the obtained expressions are adapted using Palm theory to match the system model defined for a typical UE at $(0,0)$.
The association cell $\mathcal{C}_0$ of the typical BS is defined as $\mathcal{C}_0 = \{ \mathbf{y}\in\mathbb{R}^2 : \lVert \mathbf{y} \rVert \leq \lVert \mathbf{y}-\mathbf{t}_k \rVert, \forall \mathbf{t}_k \in \Phi_k^\text{BS}, \forall k \in \mathcal{K}_0^\text{BS}\}$.
The total load to characterize is defined as $S \triangleq \sum_{u\in \mathcal{K}_\text{UE}} \Psi_u(\mathcal{C}_0)$ where $\Psi_u(\mathcal{C}_0)$ is the number of users from user tier $u$ in the typical cell. To ensure mathematical tractability, we make a first assumption inspired by \cite[Th. 4]{190703635}:
\begin{assumption}
The association cell of the typical BS is approximated by a disk of radius $R_c$ centered at the position of the BS. Therefore, $\Psi_u(\mathcal{C}_0)$ is assumed to be equal to $\Psi_u(b(\mathbf{0},R_c))$. The disk is defined to have the same area as the association cell.
\end{assumption}
For a given radius $R_c$, the random counting measures from each tier are independent, since all PPs of UEs are independent. The PGF of the total load is thus given by the product of the individual PGFs of the number of UEs from each tier, integrated over the PDF of the radius $R_c$:
\begin{equation}
\label{eq:PGF_S}
    G_S(\theta) = \int_0^{\infty} \prod_{u \in \mathcal{K}^\text{UE}} G_{\Psi_u(b(\mathbf{0},r_c))}(\theta|r_c) f_{R_c}(r_c) \textup{d}r_c. 
\end{equation}

The distribution of the typical cell radius $R_c$ is stated in lemma \ref{lem:sec4_radius}:
\begin{lemma}
\label{lem:sec4_radius}
The radius $R_c$ of the approximated circular typical cell follows a Nakagami distribution of parameters $m=3.575$ and $\Omega=(\pi \lambda_\textup{tot}^\textup{BS})^{-1}$.
\end{lemma}
\begin{proof}
See the proof of \cite[Theorem 2]{9079449}.
\end{proof}
An expression of $G_{\Psi_u(b(\mathbf{0},r_c))}(\theta|r_c)$ must be derived for each category of tiers of UEs in order to obtain the PGF of the total load $S$. Then, the PMF of the load can be obtained by performing the inverse z-transform of its PGF which is approximated by an Inverse Discrete Fourier Transform (IDFT) for numerical computation:
\begin{equation}
\label{eq:pmf_act}
    p_\mathcal{A}(n) = \text{IDFT}\left\{G_{S}\left(e^{j2\pi n/N}\right) \right\}.
\end{equation}

\subsubsection{Derivation of the PGFs}
\label{sec:derivation_pmf}
For the sake of mathematical tractability, two additional assumptions are employed:
\begin{assumption}
Neighboring BSs have an independent activity.
\end{assumption}
\begin{assumption}
\label{ass:size_radius}
The size of the radius $R_c$ is independent of the UEs' locations.
\end{assumption}
In the next lemmas, the PGFs for all categories of UEs are presented. 
\begin{lemma}
\label{lem:sec4_pgf1_2}
The PGF $G_{\Psi_u(b(\mathbf{0},r_c))}(\theta|r_c)$ for a PP $\Phi_u^\textup{UE}$ of category 1 or 2 independent of the typical BS location, is given by
\begin{align*}
& G_{\Psi_u(b(\mathbf{0},r_c))}(\theta|r_c)
= \exp\Bigg\{ -2\pi \lambda_{p,u}^\textup{UE} \\
 & \int_{\chi}^{\infty} \bigg( 1-\exp\Big[-\overline{m}_u(1-\theta)\xi(r_c,w) \Big] \bigg) w \textup{d}w \Bigg\} \numberthis \label{eq:chap4_G_Psi},
    \end{align*}
    in which
    $\xi(r,w)$ can be found in \cite[Corollary 1]{9079449} for \textcolor{\CommentsColor}{an MCP}, and $\chi$ is an excluding distance to the cluster centers:
    \begin{equation*}
    \begin{cases}
    &\text{For a PP of category 1: }\chi = 0,\\ 
    &\text{For a PP of category 2: }\chi \leq 2r_c.
    \end{cases}
    \end{equation*}
\end{lemma}
\begin{proof}
See Appendix \ref{app_A}.
\end{proof}
The value of $\chi$ for PPs of category 2 is a parameter to tune. The optimal choice is discussed in Section \ref{sec:simu}.
\newline

\begin{lemma}
\label{lem:sec4_pgf_linked}
The PGF $G_{\Psi_u(b(\mathbf{0},r_c))}(\theta|r_c)$ for a daughter PP $\Phi_{d,u}^\textup{UE}$ of category 2 whose distribution is directly linked to the position of the typical BS is given by
\begin{equation}
G_{\Psi_u(b(\mathbf{0},r_c))}(\theta|r_c) =  \exp\Big[-\overline{m}_u(1-\theta) \xi(r_c,0) \Big]
\end{equation}
in which \textcolor{\CommentsColor}{$\xi(r,0)=\frac{\min(r^2,r_{m,u}^2)}{r_{m,u}^2}$ for an MCP}.
\end{lemma}

\begin{proof}
Similar to Appendix \ref{app_A}.\\
\end{proof}

\begin{lemma}
\label{lem:sec4_pgf_3}
The PGF $G_{\Psi_u(b(\mathbf{0},r_c))}(\theta|r_c)$ for a PP $\Phi_u^\textup{UE}$ of category 3 is given by
\begin{equation}
G_{\Psi_u(b(\mathbf{0},r_c))}(\theta|r_c) =  \exp\Big[-\pi \lambda_u^\textup{UE} (1-\theta) r_c^2 \Big].
\end{equation}
\end{lemma}

\begin{proof}
The proof is obtained by applying in \eqref{eq:app_loadPCP1} of Appendix \ref{app_A} the PGF of an HPPP defined in \cite{PPP_GEN}.
\end{proof}

\subsubsection{Adaptation regarding Palm Theory}
Palm theory states that the distribution for a typical BS holds for any BS selected uniformly at random among all BSs of a PP \cite{PPP_GEN}. However, evaluating the performance of the network in Section \ref{sec:perf} requires the distribution of the activity of the $i^{th}$ closest BS to the typical UE. This is not a random selection and Palm theory does not apply. In this paper, we thus make an additional assumption.
\begin{assumption}
The load developed for the typical BS is assumed to be valid for all BSs except for the closest to the typical UE, for which the adaptation below is performed.\\
\end{assumption}
With the assumption of a circular Voronoï region, the load of the closest BS is not independent of its distance to the typical UE. Indeed, the radius $R_c$ of the Voronoï region of the closest BS must at least be larger than its distance to the typical UE. If it is not the case, this BS would not be the closest one by contradiction.

To adapt the load for the nearest BS which is from tier $k^{*}\in \mathcal{K}_0^\text{BS}$ and located at a distance $\textcolor{\CommentsColor}{r_{k^{*}}}$, the lower limit of the integration on $R_c$ of \eqref{eq:PGF_S} must be  $\textcolor{\CommentsColor}{r_{k^{*}}}$ instead of $0$. The PDF of the radius $R_c$ must in that case be normalized and is given by
\begin{equation}
    f_{R_c|\textcolor{\CommentsColor}{r_{k^{*}}}}(r_c) = \frac{f_{R_c}(r_c)}{\overline{F}_{R_c}(\textcolor{\CommentsColor}{r_{k^{*}}})}, \quad r_c \geq \textcolor{\CommentsColor}{r_{k^{*}}},
\end{equation}
where $\overline{F}_{R_c}$ is the CCDF of $R_c$ given by the regularized upper incomplete gamma function $Q\left(c,\frac{c}{\Omega}\textcolor{\CommentsColor}{r_{k^{*}}^2}\right)$ with $\Omega=(\pi \lambda_\text{tot}^\text{BS})^{-1}$.
\subsection{Distribution of the Serving Distance}
\label{sec:distances}
The expression of the distribution of the distance $\textcolor{\CommentsColor}{R_{j}}$ between the typical UE and the closest BS of tier $j$ is developed in this subsection. The corresponding PDF and CCDF are stated in Lemma \ref{lem:sec5_dist}. Lemma \ref{lem:prob_closest} then gives the expression of the term $\mathbb{P}_\text{closest}(k^{*} | \textcolor{\CommentsColor}{r_{k^{*}}})$.
\color{\CommentsColor}
\begin{lemma}
\label{lem:sec5_dist}
The PDF $f$ and the CCDF $\overline{F}$ of the distance between the typical UE and the closest BS from tier $j$, given an excluding radius $\kappa$ defining a region without BSs, are expressed as follows:
\begin{itemize}
    \item \underline{For $j\in \mathcal{K}^\textup{BS}$ and $r_{j} \geq \kappa$:}
    \begin{align*}
        f_{R_{j}}(r_{j}|\kappa) &= 2\pi \lambda_j^\text{BS} r_{j}
         \exp(-\lambda_j^\text{BS} \pi (r_{j,i}^2 - \kappa^2)). \numberthis \\
        \overline{F}_{R_{j}}(r_{j}|\kappa) &= \exp(-\lambda_j^\text{BS} \pi (r_{j}^2 - \kappa^2)). \numberthis
    \end{align*}
    \item  \underline{For $j=0$ and $\kappa \leq r_{0} \leq r_{m,0}$:}
    \begin{align}
        f_{R_{0}}(r_{0}|\kappa) & = \frac{2 r_{0}}{r_{m,0}^2-\kappa^2}. \numberthis \label{eq:dist_MCP} \\
        \overline{F}_{R_{0}}(r_{0}|\kappa) & = \frac{r_{m,0}^2-r_{0}^2 }{r_{m,0}^2-\kappa^2}.
    \end{align}
\end{itemize}
\end{lemma}
\color{black}

\begin{proof}
Taking the exclusion distance into account, the expressions are adapted from \cite{PPP_GEN} and \cite{9079449} for $j\in \mathcal{K}^\text{BS}$ and $j=0$ respectively.
\end{proof}
\textcolor{\CommentsColor}{In this paper, the PDF and the CCDF of the distance to the closest BS of subtier $\Tilde{\Phi}_{j}^{\textup{BS}}$ will be denoted $f_{\Tilde{R}_{j}}$ and $\overline{F}_{\Tilde{R}_{j}}$. Using the independent thinning property of Poisson PPs, it can be obtained by replacing $\lambda_j^\textup{BS}$ by $q_j \lambda_j^\textup{BS}$.}
\begin{lemma}
\label{lem:prob_closest}
The probability that the closest BS is from tier $k^{*}$, knowing the distance $\textcolor{\CommentsColor}{r_{k^{*}}}$ to the closest BS of this tier, is given by
\begin{equation}
    \mathbb{P}_\textup{closest}(k^{*} | \textcolor{\CommentsColor}{r_{k^{*}}}) = \prod_{j \in \mathcal{K}_0^\textup{BS} \backslash \{k^{*}\}} \overline{F}_{\textcolor{\CommentsColor}{R_{k^{*}}}}(\textcolor{\CommentsColor}{r_{k^{*}}}|R_\textup{min}).
\end{equation}
\end{lemma}

\begin{proof}
This lemma directly follows from \cite[Lemma 1]{7511509}.
\end{proof}
\subsection{Coverage and Success Probability}
\label{sec:perf}
The coverage and the success probability are derived in this section. Assuming perfect alignment of the beam for the served UE at the origin, the expression of the Signal-to-Interference-plus-Noise Ratio (SINR) for the downlink transmission from an active BS of tier $j$ at distance $r_{j}$, is given by
\begin{equation}
\label{eq:chap5_SINR}
    \text{SINR}(r_j|S_{\Phi_{j}}) = \frac{A_{j} |\rho_j|^2 r_j^{-\alpha} \mathbbm{1}(r_j\leq R_\text{max})}{\sigma_\textup{noise}^2 + \mathcal{I}^\textup{tot}(r_j)},
\end{equation}
in which $\rho_j$ denotes $\rho_{\mathbf{x}_j}$, the total interference is $\mathcal{I}^\textup{tot}$ and $A_{j} = \beta M_{j} P_{j}$. At this point, a first observation is that \eqref{eq:chap5_SINR} does not depend on the index of the serving BS. Therefore, a per-tier success probability can be derived by following the steps of the coverage analysis in \cite{7913628}. When Nakagami-m fading is assumed, the success probability is given by
\begin{equation}
\label{eq:chap5_probsuc}
    \mathbb{P}_\text{suc}^{j}(\tau,r_j)= \sum_{l=0}^{m-1} \frac{(-\tau m  r_j^{\alpha})^l}{l! (A_{j})^l } \mathcal{L}_{N,\mathcal{I}^\textup{tot}}^{(l)}\left(\frac{\tau m r_j^{\alpha}}{A_{j}} ; r_j \right),
\end{equation}
in which $r_j \leq R_\textup{max}$ and $\mathcal{L}_{N,\mathcal{I}^\textup{tot}}^{(l)}$ is the $l^{th}$ derivative of the LT of the noise and interference term $\sigma_\textup{noise}^2+\mathcal{I}^\textup{tot}$. This LT is defined as $\mathcal{L}_{N,\mathcal{I}^\textup{tot}}(s|r_j)= e^{-s\sigma_{noise}^2} \mathbb{E}[e^{-s\mathcal{I}^\textup{tot}(r_j)} ]$. \\

\subsubsection{Characterization of the Interference}
\label{sec:charact_inter}
\color{\CommentsColor}
Each interfering tier $k$ is denoted $\Tilde{\Phi}_{k}^\text{BS}$ and is obtained by applying thinning operations of probabilities $q_k$. The interference $\mathcal{I}_{(k)}$ from the active BSs of the $k^{th}$ tier when the typical UE is served by a BS of tier $j$ is thus given by
\begin{equation}
\label{eq:LT_adapt}
    \mathcal{I}_{k}(r_{j}) = \sum_{\mathbf{x}_k \in \Tilde{\Phi}_{k}^\text{BS} \setminus b(\mathbf{0},r_{j})}
      \beta P_{k} g_{\mathbf{x}_k} \lVert \mathbf{x}_k \rVert^{-\alpha} \mathbbm{1}(r_k\leq R_\text{max}),
\end{equation}
\color{black}
in which $b(\mathbf{0},r_{j})$ is the disc centered at the typical user with no active BS interfering. The total interference $\mathcal{I}^\textup{tot}$ is thus obtained by summing the contributions from each tier. To further develop the success probability, the LT of the noise and interference term is given in Corollary \ref{cor:sec5_Lnoise}, while Lemma \ref{lem:sec5_LI} states the LT of the interference term alone.
\begin{lemma}
\label{lem:sec5_LI}
The LT of the total interference term, when the serving BS is from tier $j$ and located at a distance $r_j$ of the typical UE, is given by
\color{\CommentsColor}
\begin{align}
    \mathcal{L}_{\mathcal{I}^\textup{tot}}(s|r_j) 
    & = \mathcal{L}_{\mathcal{I}_{0}}(s|r_j) \prod_{k\in \mathcal{K}^\textup{BS}} \mathcal{L}_{\mathcal{I}_{k}}(s|r_j),
\end{align}\color{black}
\textcolor{\CommentsColor}{with the LT of the interference from subtier $\Tilde{\Phi}_{k}^\textup{BS}$ given by $\mathcal{L}_{\mathcal{I}_{k}}(s|r_j) = \exp\left(-\zeta_{k}(s|r_j) \right)$.} The expression of the term $\zeta_{k}(s|r_j)$ is given by
\begin{align*}
    \zeta_{k}(s|r_j) & = \frac{2\pi q_k \lambda_k^\textup{BS}}{M_{k}} \Bigg\{r_j^2 \left[\mathcal{J}_0\left(\frac{-s A_{k} r_j^{-\alpha}}{m} \right) -1 \right] \\
    & -R_\textup{max}^2 \left[\mathcal{J}_0\left(\frac{-s A_{k} R_\textup{max}^{-\alpha}}{m} \right) -1 \right]\Bigg\}, \numberthis 
\end{align*}
where $\mathcal{J}_k(x) \triangleq {}_3F_2\left(k+\frac{1}{2},k-\nu,k+m;k+1,k+1-\nu;x\right)$, with the generalized hypergeometric function defined as ${}_p F_q(a_1,...,a_p;b_1,...,b_q;z)$  and $\nu \triangleq \frac{2}{\alpha}$.

For the interference from tier $k=0$, if $j\neq 0$ and the typical UE is from a tier of category 2, then it is not served by its coupled BS. The LT of the interference from this BS is given by
 \color{\CommentsColor}
\begin{align*}
    \mathcal{L}_{\mathcal{I}_0}(s|r_j) & = (1-q_0) -  \frac{2 q_0}{M_{0}} \Bigg[ \mathcal{J}_0\left(\frac{-s A_{0} r_j^{-\alpha}}{m}\right) \\
     & r_j^2 - \mathcal{J}_0\left(\frac{-s A_{0} R_\textup{top}^{-\alpha}}{m}\right)R_\textup{top}^2 \Bigg], \numberthis
\end{align*}
for $r_j<R_\textup{top}$ with $R_\textup{top} = \min(r_{m,0},R_\textup{max})$. Otherwise, $\mathcal{L}_{\mathcal{I}_0}(s|r_j)=1$.
\color{black}
\end{lemma}

\begin{proof}
See Appendix \ref{app_B}.
\end{proof}

\begin{corollary}
\label{cor:sec5_Lnoise}
The LT of the noise and interference term $\sigma_\textup{noise}^2 + \mathcal{I}^\textup{tot}$ is given by
\begin{equation}
    \mathcal{L}_{N,\mathcal{I}^\textup{tot}}(s|r_j) 
     = \mathcal{L}_{\mathcal{I}_0}(s|r_j) \mathcal{L}_{\mathcal{I}_{exp}}(s|r_j),
\end{equation}
where $\mathcal{L}_{\mathcal{I}_{exp}}(s|r_j)=\exp\left( \zeta_{exp}(s|r_j) \right)$ \\
with $\zeta_{exp}(s|r_j) \triangleq -s \sigma_\textup{noise}^2 - \sum_{k \in \mathcal{K}^\textup{BS}} \zeta_{k}(s|r_j).$
\end{corollary}

\subsubsection{Success Probability}
Using the lemmas from previous subsections, the probability of success stated in \eqref{eq:chap5_probsuc} can be developed. The expression when the typical UE is not from a tier of category 2 is stated in Theorem \ref{theo:sec5_psuc1}. For a typical UE from a tier of category 2, if $j=0$, then the same theorem gives the correct expression, while Theorem \ref{theo:sec5_psuc2} states the probability for $j \neq 0$. In the expression of the AAKCP metric defined in \eqref{eq:AAKCP_u}, the success probability can thus be computed using either Theorem \ref{theo:sec5_psuc1} or Theorem \ref{theo:sec5_psuc2}. 

\begin{figure*}[!t]
\normalsize
\setcounter{MYtempeqncnt}{\value{equation}}
\setcounter{equation}{31} 
\color{\CommentsColor}
\begin{align*}
c_l = & -\mathbbm{1}(l\leq1) (-1)^l \sigma_\textup{noise}^2 \frac{\tau m r_j^\alpha}{A_{j,n}} - \sum_{k \in \mathcal{K}^\textup{BS}} \mathbbm{1}(l=0) \frac{2\pi q_k \lambda_k^\textup{BS}}{\textcolor{\CommentsColor}{M_{k}} \left(R_\textup{max}^2-r_j^2 \right)^{-1}} - \sum_{k \in \mathcal{K}^\textup{BS}} \frac{4\sqrt{\pi} q_k \lambda_k^\textup{BS} \tau^l}{(l!)^2 \textcolor{\CommentsColor}{M_{k}}}  \\
& \frac{\Gamma(m+l)}{\Gamma(m)} \frac{\Gamma(l+\frac{1}{2})}{(2-\alpha l)} \left(\frac{A_{k}}{A_{j}}\right)^l \left[r_j^2 \mathcal{J}_l\left(-\tau \frac{A_{k}}{A_{j}}\right) - R_\textup{max}^2 \left(\frac{r_j}{R_\textup{max}} \right)^{\alpha l} \mathcal{J}_l\left(-\tau \frac{A_{k}}{A_{j}} \left(\frac{r_j}{R_\textup{max}}\right)^\alpha \right)\right]. \numberthis \label{eq:c_l}
\end{align*} 
\color{black} 
\setcounter{equation}{\value{MYtempeqncnt}}
\hrulefill
\vspace*{4pt}
\end{figure*}
\addtocounter{equation}{1}

\begin{theorem}
\label{theo:sec5_psuc1}
Let us consider a typical UE served by an active BS of tier $\Tilde{\Phi}_{j}^\textup{BS}$ located at distance $r_j$. This typical UE is not from a tier of category 2 or the typical UE is from a tier of category 2 but $j=0$. Then, the success probability $\mathbb{P}_\text{suc}^{j}\left(\tau,r_j\right)$ is given by the one norm of a matrix exponential. The matrix is a $m\times m$ Toeplitz matrix whose first column is defined by $[c_0, \cdots, c_{m-1}]$ with $c_l$ defined in \eqref{eq:c_l}.
\end{theorem}

\begin{proof}
If the typical UE is not from a tier of category 2 or if $j=0$, then $\mathcal{L}_{\mathcal{I}_0}(s|r_j)=1$. Therefore, the LT of the noise and interference term given in Corollary \ref{cor:sec5_Lnoise} is an exponential $\mathcal{L}_{N,\mathcal{I}^\textup{tot}}(s|r_j) = \exp[\zeta_{exp}(s|r_j)]$. We can use the result of \cite[Theorem 1]{7913628} for ad hoc networks by identifying $s=\frac{\tau m r_j^{\alpha}}{A_{j}}$, $\eta(s)$ with $\zeta_{exp}(s|r_j)$ and by computing the $l^{th}$ derivative of $\zeta_{exp}(s|r_j)$ which can be obtained from the $l^{th}$ derivative of a hypergeometric function.
\end{proof}

\begin{figure*}[!b]
\normalsize
\setcounter{MYtempeqncnt}{\value{equation}}
\setcounter{equation}{36} 
\hrulefill
\color{\CommentsColor}
\begin{align*}
\zeta_{k}^{(l)}(s;r_j) = \mathbbm{1}_{(l=0)} \frac{2\pi q_k \lambda_k^\textup{BS}}{M_{k}}\left(R_\textup{max}^2-r_j^2 \right)
        + \frac{4\sqrt{\pi}q_k \lambda_k^\textup{BS}}{l! M_{k}} \frac{\Gamma(m+l)}{\Gamma(m)} \frac{\Gamma(l+\frac{1}{2})}{(2-\alpha l)} \\
        \left(\frac{-A_{k}}{m}\right)^{l} \left[r_j^{2-\alpha l} \mathcal{J}_l\left(\frac{-sA_{k} r_j^{-\alpha}}{m}\right) - R_\textup{max}^{2-\alpha l} \mathcal{J}_l\left(\frac{-sA_{k} R_\textup{max}^{-\alpha}}{m}\right)\right].
 \numberthis \label{eq:zeta_l}
\end{align*}  
\color{black}
\setcounter{equation}{\value{MYtempeqncnt}}
\vspace*{4pt}
\end{figure*}

\begin{theorem}
\label{theo:sec5_psuc2}
Let us consider a typical UE served by an active BS of tier $\Tilde{\Phi}_{j}^\textup{BS}$ located at distance $r_j$. This typical UE is from a tier of category 2 and $j\neq0$. Then the success probability is given by
\begin{multline}
    \mathbb{P}_\text{suc}^{j}\left(\tau,r_j\right) = \sum_{l=0}^{m-1} \frac{(-\tau m  r_j^{\alpha})^l}{l! (A_{j})^l }\\
     \sum_{L=0}^{l} \binom{L}{k} \mathcal{L}_{\mathcal{I}_0}^{(l-L)}\left(\frac{\tau m r_j^{\alpha}}{A_{j}}\bigg|r_j \right) \mathcal{L}_{\mathcal{I}_{exp}}^{(L)}\left(\frac{\tau m r_j^{\alpha}}{A_{j}}\bigg|r_j \right),
\end{multline}
where:
\begin{itemize}
    \item the $l^{th}$ derivative of $\mathcal{L}_{\mathcal{I}_0}(s|r_j)$ is given by
    \color{\CommentsColor}
\begin{multline}
        \mathcal{L}_{\mathcal{I}_0}^{(l)}(s|r_j) = \mathbbm{1}_{(l = 0)} (1-q_0) - \frac{4 q_0}{l! \sqrt{\pi} M_{k}} \frac{\Gamma(m+l)}{\Gamma(m)} \\ 
        \frac{\Gamma(l+\frac{1}{2})}{(2-\alpha l)}
        \left(\frac{-A_{0}}{m}\right)^{l} \Bigg[r_j^{2-\alpha l} \mathcal{J}_l\left(\frac{-sA_{0} r_j^{-\alpha}}{m}\right) -\\
         R_\textup{top}^{2-\alpha l} \mathcal{J}_l\left(\frac{-sA_{0} R_\textup{top}^{-\alpha}}{m}\right)\Bigg],
        \end{multline}
        \color{black}
\item the $l^{th}$ derivative of $\mathcal{L}_{\mathcal{I}_{exp}}$ can be computed by recursion for $l\geq 1$ as follows:
        \begin{equation}
           \mathcal{L}_{\mathcal{I}_{exp}}^{(l)}(s|r_j) = \sum_{i=0}^{l-1} \binom{l-1}{i} \zeta_{exp}^{l-i}(s|r_j) \mathcal{L}_{exp}^{(i)}(s|r_j),
        \end{equation}
\item the $l^{th}$ derivative of the exponent term $\zeta_{exp}$ is given by
\color{\CommentsColor}
        \begin{equation}
            \zeta_{exp}^{(l)}(s|r_j) = -\mathbbm{1}_{(l\leq1)} \sigma_\textup{noise}^2 s^{1-l} - \sum_{k \in \mathcal{K}^\textup{BS}} \zeta_{k}^{(l)}(s|r_j),
        \end{equation}     
        \addtocounter{equation}{1}
        where $\zeta_{k,n'}^{(l)}(s|r_j)$ is defined in \eqref{eq:zeta_l} \color{black}
\end{itemize}
\end{theorem}

\begin{proof}
In this situation, the LT of the noise and interference term cannot be written as an exponential. The success probability is thus not given by the one norm of an exponential matrix. The $l^{th}$ derivative of $\mathcal{L}_{N,\mathcal{I}^\textup{tot}}$ in \eqref{eq:chap5_probsuc} of lemma 8 must therefore be computed using Leibniz's general derivation rule with $s=\frac{\tau m r_j^\alpha}{A_{j}}$.
\end{proof}

\subsubsection{Coverage Probability}
One last term in the AAKCP metric is not yet defined. This term is denoted as $\mathbb{P}_\text{cov}\left(\tau \Big|\textcolor{\CommentsColor}{S_{\Phi_{k^{*}}}'}, \textcolor{\CommentsColor}{r_{k^{*}}} \right)$ and refers to the coverage probability of the network when the closest BS is sleeping. The distance to this sleeping BS is known and denoted as $\textcolor{\CommentsColor}{r_{k^{*}}}$. It results in an exclusion region of serving BSs around the typical UE. \color{\CommentsColor}In order to calculate this coverage probability, the expression is split in a sum of per-tier $\Tilde{\Phi}_{j}^\text{BS}$ coverage probability:
\begin{equation}
   \label{eq:chap5_covprob}
    \mathbb{P}_\text{cov}\left(\tau \Big|\textcolor{\CommentsColor}{S_{\Phi_{k^{*}}}'}, r_{k^{*}} \right) = 
    \sum_{j\in\mathcal{K}_0^\text{BS}} \mathbb{P}_\text{cov}^{j}\left(\tau \Big| \textcolor{\CommentsColor}{S_{\Phi_{k^{*}}}'}, r_{k^{*}} \right)
\end{equation}
The coverage probability of subtier $j\in \mathcal{K}_0^\text{BS}$  must be weighted by the probability of being connected to this tier $\Tilde{\Phi}_{j}^\text{BS}$. The corresponding event is denoted as $S_{\Phi_{j}}$ so $\mathbb{P}(S_{\Phi_{j}}|r_c)$ represents the probability to be connected to subtier $\Tilde{\Phi}_{j}^\text{BS}$, knowing the distance $r_c$ to the closest awake BS from this tier. This probability is stated in Lemma \ref{lem:sec5_Pphi}. Then, Theorem \ref{theo:sec5_probcov} states the expression of the per-subtier coverage probabilities. 

\begin{lemma}
\label{lem:sec5_Pphi}
The probability to be connected to a given subtier, knowing the distance $r_c$ to the closest awake BS from this subtier, is given by
\begin{equation}
    \mathbb{P}(S_{\Phi_{j}}|r_c,\textcolor{\CommentsColor}{S_{\Phi_{k^{*}}}'}, r_{k^{*}}) =  \overline{F}_0
    \prod_{k\in \mathcal{K}^\textup{BS}} \overline{F}_{k}  \numberthis \label{eq:lemma13}
\end{equation}
with:
\begin{align*}
& \overline{F}_{k} =\begin{cases}
  \multicolumn{1}{@{}c@{\quad}}{1} & \text{ for } j = k, \\ 
  \multicolumn{1}{@{}c@{\quad}}{\overline{F}_{\Tilde{R}_{k}}(r_c|r_{k^{*}})} & \text{ otherwise }, 
\end{cases} \numberthis\\
& \overline{F}_{0} =\begin{cases}
  ~0, \text{ for } k^*=0,j=0 \\
  ~1, \text{ for } k^*=0, j\neq0, \\
  q_0 ~\mathbbm{1}_{(r_c\leq r_{m,0})}, \text{ for } k^*\neq0, j = 0, \\  
  (1-q_0) + \overline{F}_{R_{0}}(r_c|r_{k^{*}})q_0, \text{ otherwise.}
  \end{cases} \numberthis
\end{align*}
\end{lemma}

\begin{proof}
The typical UE is connected to a given tier if all BSs from all other tiers are either asleep or further away from the typical UE than the serving BS at tier $j$. This probability is thus calculated for the tier $0$ as the sum of the probability of being asleep with the probability of being further away, minus the joint probability for each BS of each tier. For the other tiers, the CCDF of the distance to the subtiers can be used.
\end{proof}

\begin{theorem}
\label{theo:sec5_probcov}
The coverage probability from a tier $j$, weighted by the probability to be connected to this particular tier, when the closest BS is asleep and located at a distance $r_{k^{*}}$ is given by
\begin{align*}
\mathbb{P}_\text{cov}^{j}\left(\tau \Big| \textcolor{\CommentsColor}{S_{\Phi_{k^{*}}}'}, r_{k^{*}} \right) & = \int_{r_{k^{*}}}^{R_\textup{max}} \mathbb{P}(S_{\Phi_{j}}|r_c,\textcolor{\CommentsColor}{S_{\Phi_{k^{*}}}'}, r_{k^{*}})\\
& \mathbb{P}_\text{suc}^{j}(\tau,r_c)f_{\Tilde{R}_{j}}(r_c|r_{k^{*}})~\textup{d}r_c. \numberthis 
\end{align*}
\end{theorem}

\begin{proof}
This expression follows from the definition coverage probability of the sub-tier $\Tilde{\Phi}_{j}$ weighted by the probability to be connected to this sub-tier.
\end{proof}
\color{black}

\color{\CommentsColor}
\subsection{Asymptotic Expressions}
\label{sec:asymptotic}
This section derives asymptotic expressions for the probability of success stated in Theorem \ref{theo:sec5_psuc1} and \ref{theo:sec5_psuc2}. The special case of Rayleigh fading (i.e. $m=1$) and no blockage effect (i.e. $R_\textup{max} \to +\infty$) will be considered, which are the two assumptions made in the works of Table \ref{tab:SotA} that do not study mm-Wave. In this case, the probability of success expression stated in \eqref{eq:chap5_probsuc} simplifies to
\begin{equation}
\mathbb{P}_\text{suc}^{j}(\tau,r_j)=  \mathcal{L}_{N,\mathcal{I}^\textup{tot}}\left(\frac{\tau m r_j^{\alpha}}{A_{j}} ; r_j \right).
\end{equation} 
The probability of success for the asymptotic case is stated in Theorem \ref{theo:sec5_psuc_asymp}, using this definition.

\begin{figure*}[!t]
\normalsize
\setcounter{MYtempeqncnt}{\value{equation}}
\setcounter{equation}{44} 
\color{\CommentsColor}
\begin{align*}
\eta\left(\frac{\tau r_j^\alpha}{A_{j}}\bigg|r_j\right) 
& = \frac{-\tau r_j^\alpha \sigma_\textup{noise}^2}{A_{j}} - \sum_{k \in \mathcal{K}^\textup{BS}} \frac{2\pi q_k \lambda^\textup{BS}_k}{M_{k}} r_j^2 \left[\mathcal{F}\left(-\tau \frac{A_{k}}{A_{j}} \right)  -1\right] \numberthis \label{eq:eta}\\
& \overset{\alpha=4}{=} \frac{-\tau r_j^4 \sigma_\textup{noise}^2}{A_{j}} - \sum_{k \in \mathcal{K}^\textup{BS}} \frac{2\pi q_k \lambda^\textup{BS}_k}{M_{k}} r_j^2 \left[\sqrt{1+\tau \frac{A_{k}}{A_{j}}}-1\right] \numberthis \label{eq:eta_alpha}\\
\mathcal{L}_{\mathcal{I}_0}\left(\frac{\tau r_j^\alpha}{A_{j}}\bigg|r_j \right) 
& = (1-q_0) + \frac{2 q_0}{M_{0} (r_{m,0}^2-r_j^2)} \left[ r_{m,0}^2 \mathcal{F}\left(-\tau \frac{A_{0} r_j^\alpha}{A_{j} r_{m,0}^\alpha}\right) - r_j^2 \mathcal{F}\left(-\tau \frac{A_{0}}{A_{j}}\right)\right]\numberthis \label{eq:LT_0} \\
& \overset{\alpha=4}{=} (1-q_0) + \frac{2 q_0}{M_{0} (r_{m,0}^2-r_j^2)}  \left[r_{m,0}^2 \sqrt{1+\tau \frac{A_{0} r_j^4}{A_{j} r_{m,0}^4}} - r_j^2 \sqrt{1+\tau \frac{A_{0}}{A_{j}}}~\right] \numberthis \label{eq:LT_0_alpha} 
\end{align*} 
\color{black} 
\setcounter{equation}{\value{MYtempeqncnt}}
\hrulefill
\vspace*{4pt}
\end{figure*}

\begin{theorem}
\label{theo:sec5_psuc_asymp}
Consider a typical UE served by a BS of tier $\Tilde{\Phi}_{j}^\textup{BS}$, which is located at distance $r_j$. Assuming Rayleigh fading and no blockage effect, the probability of success is calculated as follows
\begin{equation}
\mathbb{P}_\text{suc}^{j}(\tau,r_j) 
= \mathcal{L}_{\mathcal{I}_0}\left(\frac{\tau r_j^\alpha}{A_{j}}\bigg|r_j \right) \exp\left[\eta\left(\frac{\tau r_j^\alpha}{A_{j}}\bigg|r_j\right) \right],
\end{equation}
where $\eta$ is given by \eqref{eq:eta} with $\mathcal{F}(x) \triangleq {}_2F_1\left(\frac{1}{2},-\nu;1-\nu;x\right)$. For the specific case in which the path loss exponent equals $\alpha=4$, the expression for $\eta$ simplifies to \eqref{eq:eta_alpha}. If $j \neq 0$, $r_{m,0}> r_j$ and the typical UE is from a tier of category 2 then $\mathcal{L}_{\mathcal{I}_0}$ is given by \eqref{eq:LT_0} and \eqref{eq:LT_0_alpha} for $\alpha=4$. Otherwise, $\mathcal{L}_{\mathcal{I}_0}=1$.
\end{theorem}
\addtocounter{equation}{4}
\begin{proof}
The development is similar to Lemma \ref{lem:sec5_LI} by considering $R_\textup{max} \to +\infty$. The expression can be simplified for $\alpha=4$ by identifying the series expansion of $\sqrt{1-x}$ in the development of the Gaussian hypergeometric function $\mathcal{F}(x)$.
\end{proof}

The effect of the number of antennas $M_k$ and the transmitting power $P_k$ on the probability of success can be observed in \eqref{eq:eta_alpha} and \eqref{eq:LT_0_alpha}. Recalling $A_k = \beta M_k P_k$, it can be observed that the sole disparity amid the influence of the two parameters lies in the term $\frac{2}{M_k}$. On the one hand, doubling the amount of power transmitted by each tier will augment the success probability by doubling the Signal-to-Noise Ratio (SNR) without any effect on the interference. On the other hand, doubling the number of antennas will also augment the SNR by the same amount but will have a different effect of the interference. Interestingly, while the interference from independent tiers decreases when $M_k$ increases, the interference from tier 0 increases with $M_0$. 

Regarding the EE, the impact of $P_k$ and $M_k$ on the power consumption given in \eqref{eq:Power_consumed} is defined by the value of the parameters $\Delta_p$ and $P_a$, with typically, $P_a < \Delta_p$. Consequently, these factors need to be considered when designing real networks.

Furthermore, it is possible to deduce the existence of an optimal quality of service threshold, $\tau$, from the asymptotic expressions. As $\tau$ increases, the probability of success and, subsequently, the AAKCP metric decrease. However, in the expression for the EE given in \eqref{eq:EE}, the AAKCP multiplies the term $\log_2(1+\tau)$, which increases as $\tau$ increases. Therefore, a trade-off must be defined between the two effects.

\color{black}
\section{Numerical Results}
\label{sec:simu}
The validity of the analytical model for the load is first assessed through MC simulations. Then, the accuracy of the analytical expression of the
AAKCP metric is verified. Finally, the analytical expressions are used to study the benefit in term of EE of different strategies. For the simulations we consider a 2-tier HetNet and 3 tiers of UEs defined by:
\color{\CommentsColor}
\begin{itemize}
    \item an independent tier of BS with intensity $\lambda_1^\text{BS} = 10^{-4}~\text{m}^{-2}$;
    \item a UE-dependent tier of BS with intensity $\lambda_2^\text{BS} =  2.5\cdot10^{-5}~\text{m}^{-2}$; the dependent tier of UEs of category 2 is distributed as a MCP with $\overline{m}_2 = 40$ and $r_{m,2}=20~\text{m}$;
    \item an independent tier of UEs of category 1 distributed as a MCP with $\lambda_{p,3}^\text{UE}=10^{-4}$, $\overline{m}_3 = 10$ and $r_{m,3}=20~\text{m}$;
    \item an independent tier of UEs of category 3 with intensity  $\lambda_4^\text{UE}=10^{-3}~\text{m}^{-2}$.
\end{itemize}
\color{black}
Unless otherwise mentioned, the other parameters used for the simulation are given in Table \ref{tab:chap4_param_net}. \textcolor{\CommentsColor}{These parameters are chosen to match the power consumption model published by Ericsson in \cite{9013130}. They are defined for a network with 2 tiers of macro and street macro BSs in an urban scenario.}

Note that the same power parameters are considered for tiers $1$ and $2$. 

\color{\CommentsColor}
\begin{table}
\renewcommand{\arraystretch}{1.3}
\caption{Network parameters}
\label{tab:chap4_param_net}
\centering
\color{\CommentsColor}
\begin{tabular}{|c|c||c|c|}
\hline
Parameters & Values & Parameters & Values \\ \hline \hline
$m$ & $1$ & $R_\text{min},~R_\text{max}$& $1,~400~\text{m}$\\ \hline
$\beta, ~\alpha$& $1, ~4$ &$P_1, ~P_2$ & $43, ~43~\text{dBm}$    \\ \hline
$\sigma_{noise}^2$ & $3~10^{-2}$ &$M_1, ~M_2$ & $64, ~64$\\ \hline
$\tau$& $5~\text{dB}$ & $P_\text{stat}, ~P_\text{sleep}$       & $260,~75~\text{W}$ \\ \hline
$P_a, ~\Delta_p$& $1~\text{W}, ~4$ & $\chi, ~\frac{d}{v}$       & $2r_c, ~0.5$ \\ \hline
\end{tabular}
\end{table}
\color{black}

\subsection{Accuracy of the Analytical Expressions}
\begin{figure*}[!t]
\centering
\subfloat[Tier 1]{\includegraphics[width = 0.49\textwidth]{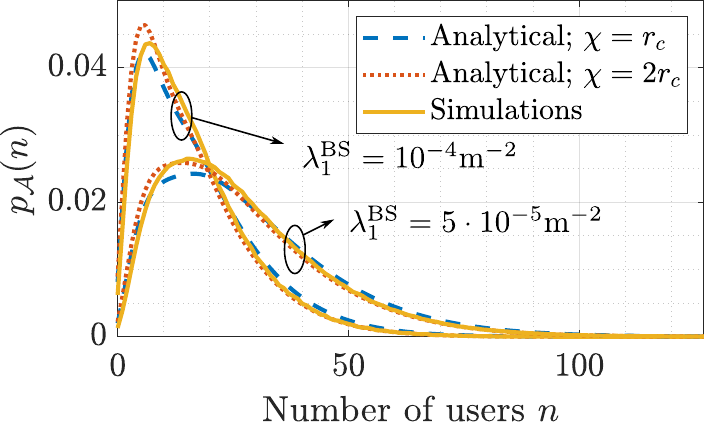}}
\label{fig:graphe1a}
\hfill
\subfloat[Tier 2]{\includegraphics[width = 0.49\textwidth]{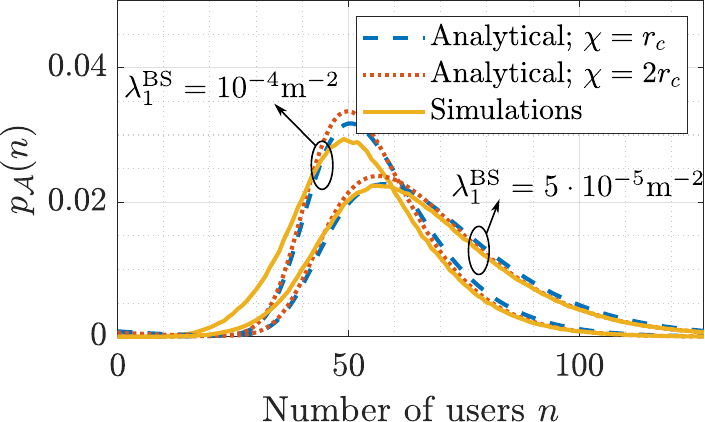}}
\label{fig:graphe1b}
\caption{\textcolor{\CommentsColor}{Accuracy of the load distribution for different densities of BSs $\lambda_1^\text{BS}$. On the left (resp. right), we can see the PMF for BSs without (resp. with) dependent users.}}
\label{fig:graphe1}
\end{figure*}

First, we can observe the accuracy of the PMF of the load for the BS of tier $1$ and $2$ respectively on \figurename \ref{fig:graphe1a} and \ref{fig:graphe1b}. The analytical expression is evaluated for two different values of the parameter $\chi$ which denotes the excluding distance to the cluster centers of category 2. With the circular Voronoï region assumption, there should not be any BS within a distance which is twice the radius of the circular Voronoï region (denoted as $r_c$). There should thus be an excluding distance $\chi=2r_c$ for the cluster centers of tier $2$. The circular region assumption is strong so it can be interesting to consider $\chi<2r_c$ to increase the accuracy of the distribution of the load for the BS of tier 1. However, in this paper, $\chi=2r_c$ is considered for the next simulations to be consistent with the circular region assumption.

\begin{figure}
\centering
\includegraphics[width=.99\linewidth]{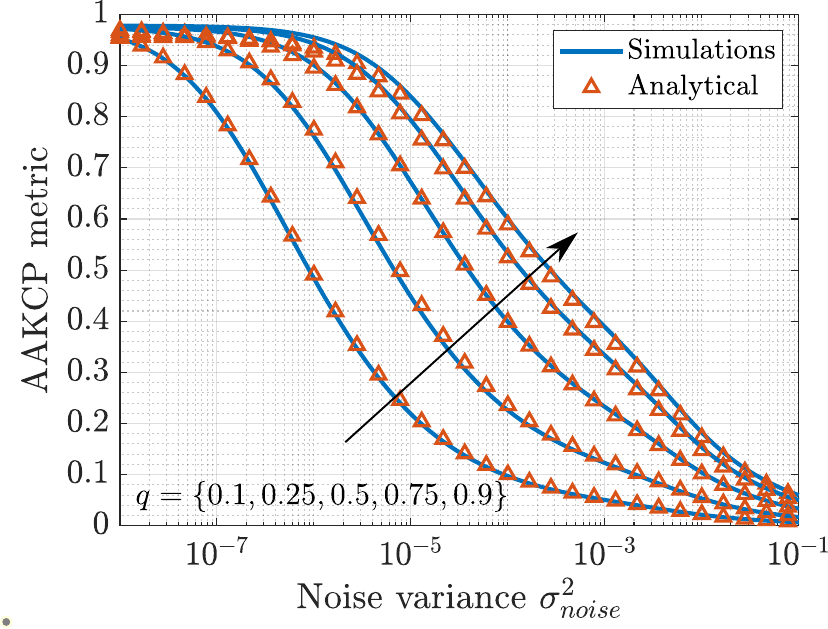}
\caption{\textcolor{\CommentsColor}{Accuracy of AAKCP metric for different ratios of active BSs.}}
\label{fig:graphe2}
\end{figure}

Next, \figurename \ref{fig:graphe2} confirms the accuracy of the AAKCP for an on/off power scheme, despite the numerous assumptions.
It can be observed that the coverage is higher when the ratio of active BSs increases. Indeed, when the ratio increases, the distance between the awake BSs and its served users decreases. The path loss hence decreases. Thanks to the use of directional arrays, the power of the useful signal increases faster than the power of the interference. Therefore, it results in an increasing coverage.

\color{\CommentsColor}
\subsection{Analysis of the System Parameters}
\color{black}

\begin{figure}
\centering
\includegraphics[width=.99\linewidth]{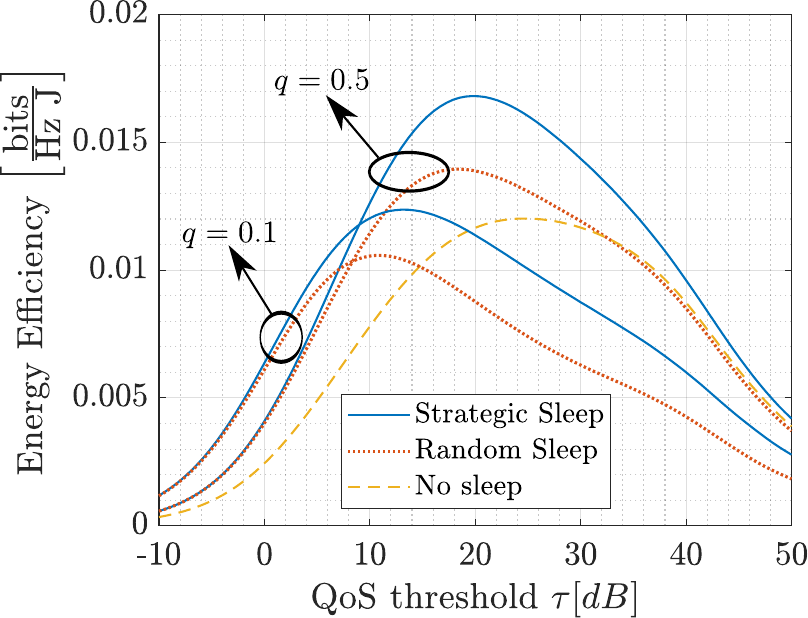}
\caption{\textcolor{\CommentsColor}{Energy efficiency for different sleeping strategies as a function of the QoS threshold $\tau$.}}
\label{fig:graphe3}
\end{figure}
\figurename \ref{fig:graphe3} illustrates the evolution of the EE with respect to the QoS threshold. Different types of strategies and sleeping ratios are compared. The optimal cell load dependent sleeping strategies outperform both RS and no sleep strategies in terms of EE. We can observe that the optimal QoS threshold decreases when the ratio of active BSs decreases. This results from a trade-off between the throughput and the network energy consumption. \textcolor{\CommentsColor}{Furthermore, simulations confirm the existence of an optimal QoS threshold discussed in section \ref{sec:asymptotic} based on the derived asymptotic expressions. As the ratio of active BSs decreases, the optimal QoS also decreases due to the AAKCP metric decreasing with the ratio.}
\begin{figure}
\centering
\includegraphics[width=.99\linewidth]{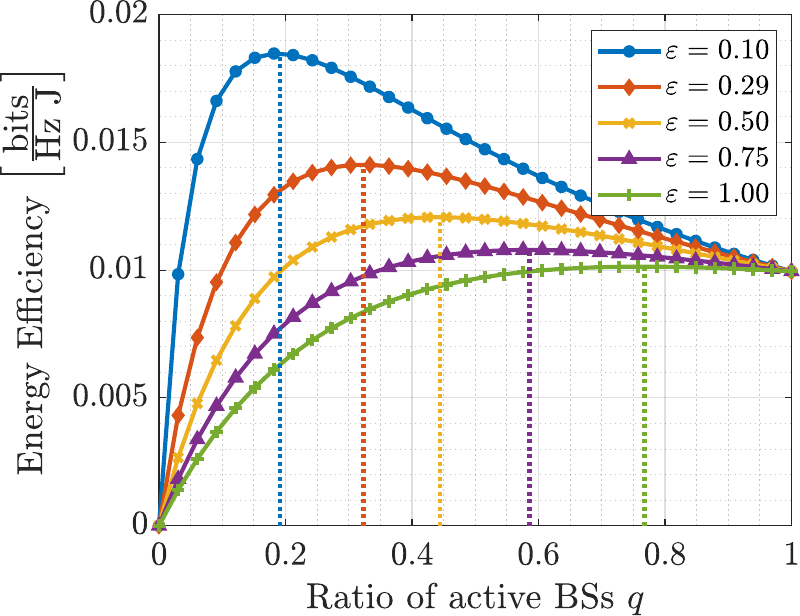}
\caption{\textcolor{\CommentsColor}{Energy efficiency for different ratios of power consumption $\epsilon = P_{sleep}/P_{stat}$ as a function of the ratio of active BSs.}}
\label{fig:graphe5}
\end{figure}
\color{\CommentsColor}

The effect on the optimal sleeping ratio of the power consumption ratio $\epsilon = P_{sleep}/P_{stat}$ can be observed on \figurename \ref{fig:graphe5}. The existence of this optimum ratio results from a balance between the decreasing distance to a serving BS and the lesser power consumption of sleeping BSs. The more efficient the power saving in sleep mode compare to the active static power consumption, the smaller the optimal ratio of active BSs. This shows that a network provider should take this ratio into account when designing its sleeping strategy. 

\begin{figure}
\centering
\includegraphics[width=.99\linewidth]{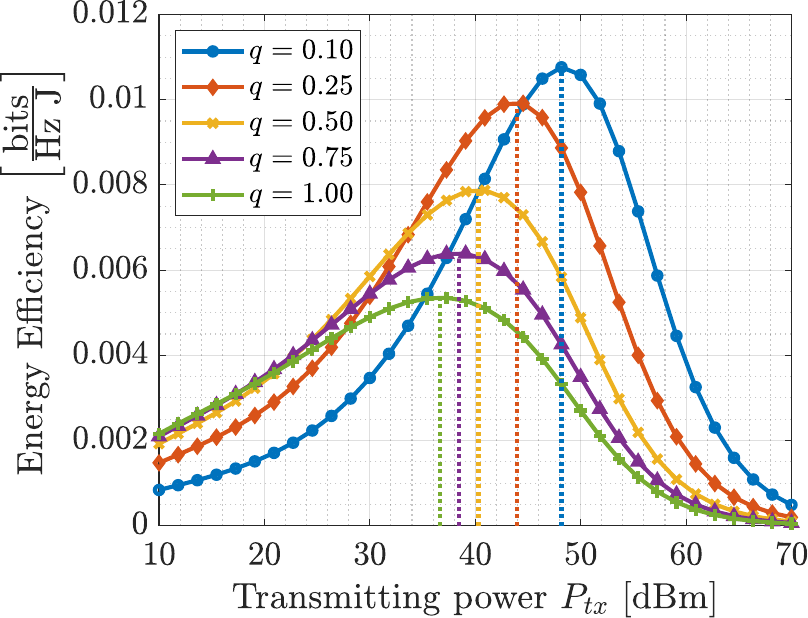}
\caption{\textcolor{\CommentsColor}{Energy efficiency for different ratios of active BSs as a function of the transmitted power.}}
\label{fig:graphe6}
\end{figure}

\figurename \ref{fig:graphe6} and \ref{fig:graphe7} depict the evolution of the EE with respect to  the transmitting power $P$ and the number of antennas $M$. As deduced from the discussion in subsection \ref{sec:asymptotic} based on the analytical expressions, the optimal value of these system parameters depends on the sleeping strategy. When the ratio of active BSs decreases, the optimal transmitting power increases to balance the loss in SNR happening due to the increasing distance to serving BSs. As expected, the evolution of the EE with respect to the number of transmitting antennas does not follow the same pattern. For low number of antennas, increasing them will increase the SNR significantly to increase the EE. Once the optimal number is reached, the gain in coverage is not significant enough to balance the increased power consumption.

\begin{figure}
\centering
\includegraphics[width=.99\linewidth]{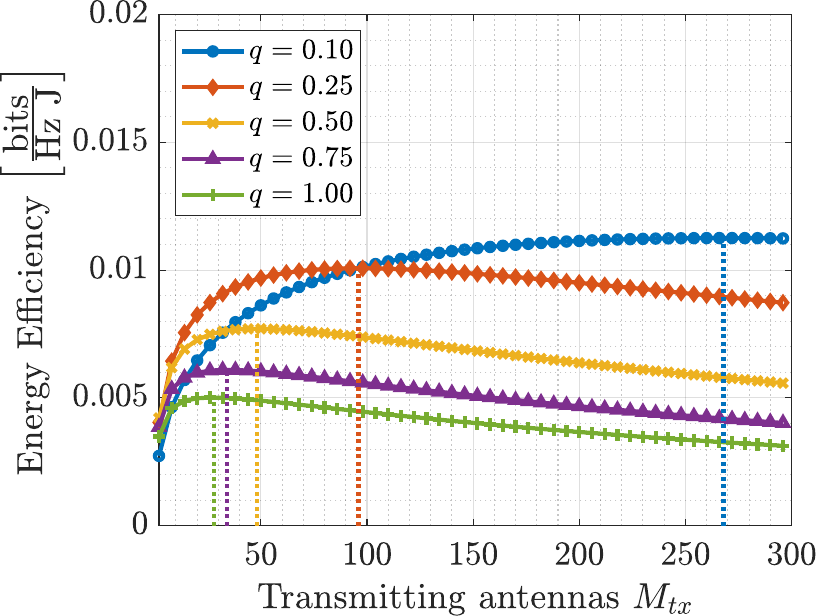}
\caption{\textcolor{\CommentsColor}{Energy efficiency for different ratios of active BSs as a function of the number of transmitting antennas.}}
\label{fig:graphe7}
\end{figure}

\color{black}

\section{Conclusion}
\label{sec:conclu}
\textcolor{\CommentsColor}{In this paper, a general and tractable framework to study the benefit of sleep mode strategies in wireless communication networks has been presented.} To this end, a new metric called the active K-tier coverage probability has been defined. To evaluate this metric, SG has been used to derive the PMF of the number of users per cell for HetNets with non-uniform user distribution. Analytical expressions for the AAKCP have been derived for on/off power control. Finally, MC simulations have been used to validate the accuracy of the obtained models which have then been used to evaluate the benefits of power control.

Possible extensions could include an average or instantaneous signal based association scheme such that UEs are able to associate with the BS providing the highest perceived average or instantaneous SINR. This requires adapting the definition of the activity to be coherent with the association scheme. \textcolor{\CommentsColor}{A potential second extension could involve expanding the system model to incorporate  coordination between BSs within the context of Coordinated MultiPoint (CoMP) transmission or cell-free massive MIMO. However, such an extension would heighten the complexity of the analytical expressions developed.}

\section*{Acknowledgments}
Computational resources have been provided by the Consortium des Équipements de Calcul Intensif (CÉCI), funded by the Fonds de la Recherche Scientifique de Belgique (F.R.S.-FNRS) under Grant No. 2.5020.11 and by the Walloon Region

\appendices
\section{} \label{app_A} 
The proof extends the development from \cite{9079449}. From its definition, the PGF for a PP $\Phi_u^\text{UE}$ of category 1 or 2 can be developed as follows:
\begin{align*}
    G_{\Psi_u(b(\mathbf{0},r_c))}(\theta|r_c) 
    & = \mathbb{E}\left[\theta^{\Psi_u(b(\mathbf{0},r_c))} \right] \\
    & = \mathbb{E}\left[\prod_{\mathbf{x}\in \Phi_u^\text{UE}} \theta^{\mathbbm{1}(\lVert \mathbf{x} \rVert \leq r_c)} \right]. \numberthis \label{eq:app_loadPCP1}
\end{align*}

Using the probability generating functional for a PCP given in \cite{8187697}, $G_{\Psi_u(b(\mathbf{0},r_c))}(\theta|r_c)$ can be written as follows:
\begin{align*}
    & G_{\Psi_u(b(\mathbf{0},r_c))}(\theta|r_c) = \exp\bigg(-\lambda_{p,u}^\text{UE} \int  1- \\
    & \exp\left(\overline{m}_u \left[\int \theta^{\mathbbm{1}(\lVert \mathbf{x}+\mathbf{y} \rVert \leq r_c)} f(\mathbf{x}) \textup{d}\mathbf{x} -1 \right] \right) \textup{d}\mathbf{y}\bigg), \numberthis \label{eq:app_loadPCP2}
\end{align*}
where $f$ is the PDF associated with the PCP (given in section \ref{sec:math_back}), $\mathbf{y}$ defines the position of the cluster center and $\mathbf{x}$ the position of the points. \eqref{eq:app_loadPCP2} can be further developed using polar coordinates such that $v=\lVert \mathbf{x} \rVert$ and $w=\lVert \mathbf{y} \rVert$:
\begin{align*}
   & G_{\Psi_u(b(\mathbf{0},r_c))}(\theta|r_c)
    = \exp\bigg(-2\pi \lambda_{p,u}^\text{UE} \int_{\chi}^{\infty} 1- \\
   & \exp\left[\overline{m}_u \int_{0}^{r_c} (\theta-1) f_d(v|w) \textup{d}v  \right]  w \textup{d}w\bigg), \numberthis \label{eq:app_loadPCP3}
\end{align*}
where $\chi$ is an excluding radius for the cluster centers of the PP and $f_d(v|w)$ is the PDF of the distance to the origin for a point of the PCP, knowing the distance between the corresponding cluster center and the origin. The expression of $f_d(v|w)$ for an MCP is stated in \cite[(3)]{9079449}.

For a PP of category 1, there is no constraint on the position of the cluster centers, so $\chi=0$. For a PP of category 2, there is one constraint as the radius of the circular Voronoï region $r_c$ is known and the users are distributed around a BS. Thus, there should be no BS within a distance $\chi = 2r_c$. However, as the circular Voronoï assumption is strong, we leave $\chi$ as a parameter of the system. Its optimal choice is discussed in Section \ref{sec:simu}. Finally, some simplifications in \eqref{eq:app_loadPCP3} finishes the proof.

\section{}\label{app_B}
In this proof, the expression of the LT of the total interference term $\mathcal{I}^\textup{tot}$ is developed. The interference from each tier being independent, the expression of the LT can be written as \textcolor{\CommentsColor}{
$\mathcal{L}_{\mathcal{I}^\textup{tot}}(s|r_j) = \prod_{k\in \mathcal{K}_0^\text{BS}} \mathcal{L}_{\mathcal{I}_{(k)}}(s|r_{j}),$
where $\mathcal{L}_{\mathcal{I}_{(k)}}$ is the LT of the interference from the tier $\Tilde{\Phi}_{k}^\text{BS}$.} This term is developed successively for $k\in \mathcal{K}^\text{BS}$ and $k=0$.

For $k\in \mathcal{K}^\text{BS}$, the position of the typical UE is completely independent from the one of the BSs of tier $k$. \textcolor{\CommentsColor}{Therefore, the position of the interfering BSs can be seen as an independent HPPP of intensity $\lambda_{\mathcal{I}_k}=q_k \lambda^\text{BS}_k$.} The LT for $k\in \mathcal{K}^\text{BS}$ can be obtained with a similar development as \cite[Appendix D]{7913628}.

For the interference from $k=0$, different cases have to be taken into account:
\begin{itemize}
    \item If the typical UE is not from a tier of category 2, there is no BS in tier $0$ so $\mathcal{L}_{\mathcal{I}_0}(s|r_j)=1$.
    \item If $j=0$, there is no interference from this term since the BS of tier $0$ is the serving one, hence $\mathcal{L}_{\mathcal{I}_0}(s|r_j)=1$ in this case as well. 
    \item If $j\neq 0$, there is interference from one BS which is the cluster center.
\end{itemize}

The LT of the interference term in the third case can be developed as follows:\color{\CommentsColor}
\begin{align*}
    \mathcal{L}_{\mathcal{I}_0}(s|r_j)
    & = (1-q_0) + q_0 \int_{r_j}^{R_\textup{top}} \\
    & \mathbb{E}_{g} \left[\exp\left(-s \beta g P_{0} r^{-\alpha} \right)\right] f_{\textcolor{\CommentsColor}{R_{0}}}(r|r_j) dr \numberthis
\end{align*}
\color{black}
where $\textcolor{\CommentsColor}{f_{R_{0}}}$ is given in \eqref{eq:dist_MCP}.

\bibliographystyle{IEEEtran}
\bibliography{IEEEabrv,References}

\begin{IEEEbiography}[{\includegraphics[width=1in,height=1.25in,clip,keepaspectratio]{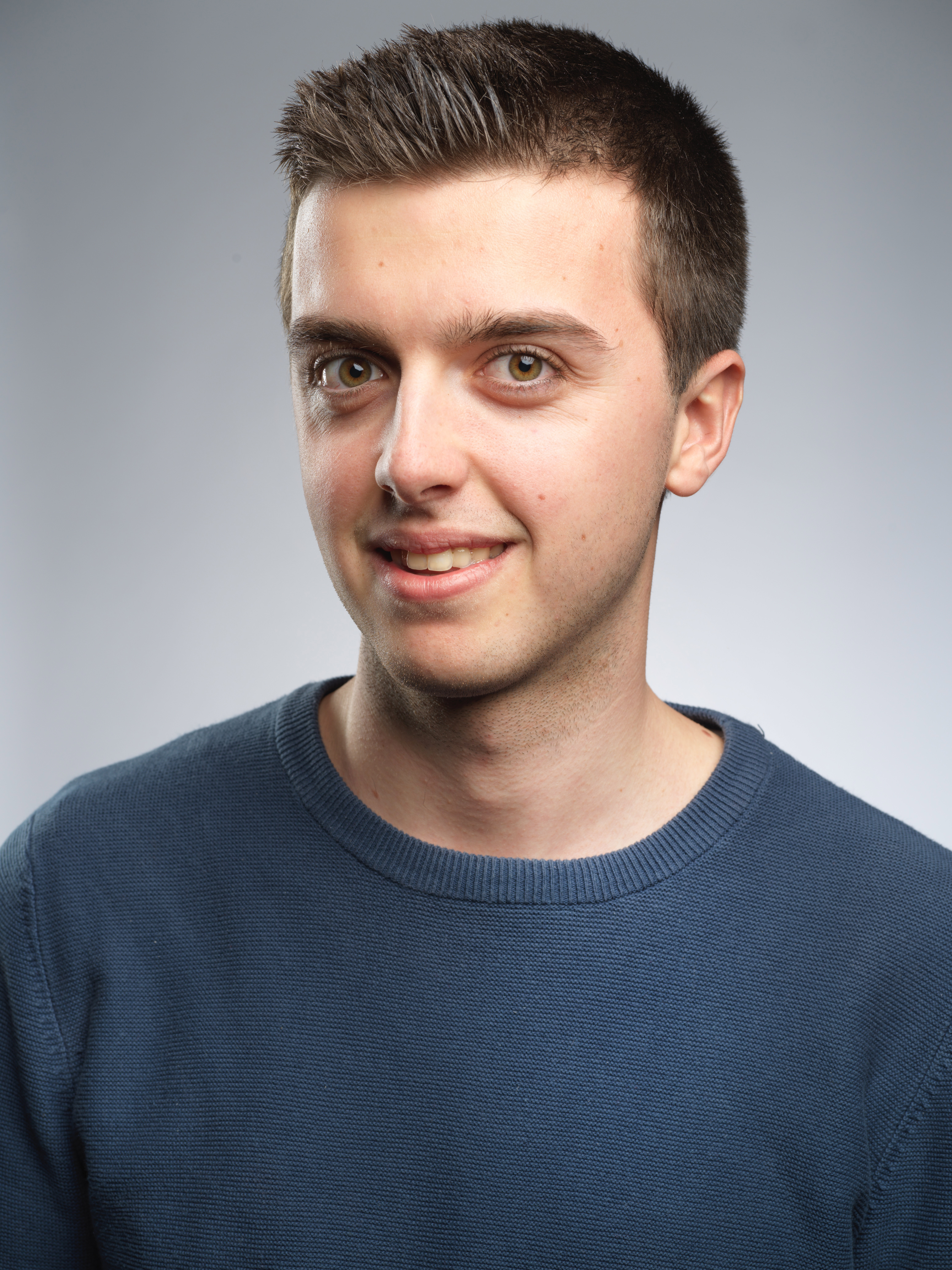}}]{Martin Willame}
(Member, IEEE) obtained his electrical engineering degree from the Université catholique de Louvain (UCLouvain), Louvain-la-Neuve, Belgium, in 2021. Since October 2021, he has been pursuing a Ph.D. at UCLouvain and ULB with a scholarship from F.N.R.S.-E.O.S. His research interests encompass sleep mode strategies, stochastic geometry, Wi-Fi-based passive radars, and multistatic systems. Additionally, in 2024, he commenced work on the European ECOsystem for greeN Electronic (EECONE) project, focusing on sustainability metrics.
\end{IEEEbiography}

\begin{IEEEbiography}[{\includegraphics[width=1in,height=1.25in,clip,keepaspectratio]{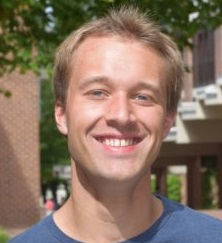}}]{Charles Wiame}
(Member, IEEE) earned his M.Sc. degree in electrical engineering from UCLouvain, Belgium, in 2017. As a Ph.D. student under the guidance of Prof. L. Vandendorpe and Prof. C. Oestges at UCLouvain, he successfully obtained his Ph.D. degree in 2023. His doctoral research focused on exploring the trade-offs between coverage and electromagnetic field (EMF) exposure in wireless systems, approached from a stochastic geometry perspective. Simultaneously, Charles served as a teaching assistant and lecturer at both the bachelor and master levels.
In 2022, he was visiting researcher in the lab of Prof. Emil Björnson at the Kungliga Tekniska högskolan (KTH), Stockholm, Sweden, where he studied cell-free massive Multiple Input Multiple Output (MIMO) systems.
After receiving a postdoctoral research fellowship from the Belgian American Education Foundation (BAEF), Charles joined the Research Laboratory of Electronics (RLE) at the Massachusetts Institute of Technology (MIT). He currently works in the Reliable Communications and Network Coding (NCRC) group, led by Prof. Muriel Médard. His ongoing research projects are dedicated to improvements in the GRAND decoder for wireless systems.

\end{IEEEbiography}

\begin{IEEEbiography}[{\includegraphics[width=1in,height=1.25in,clip,keepaspectratio]{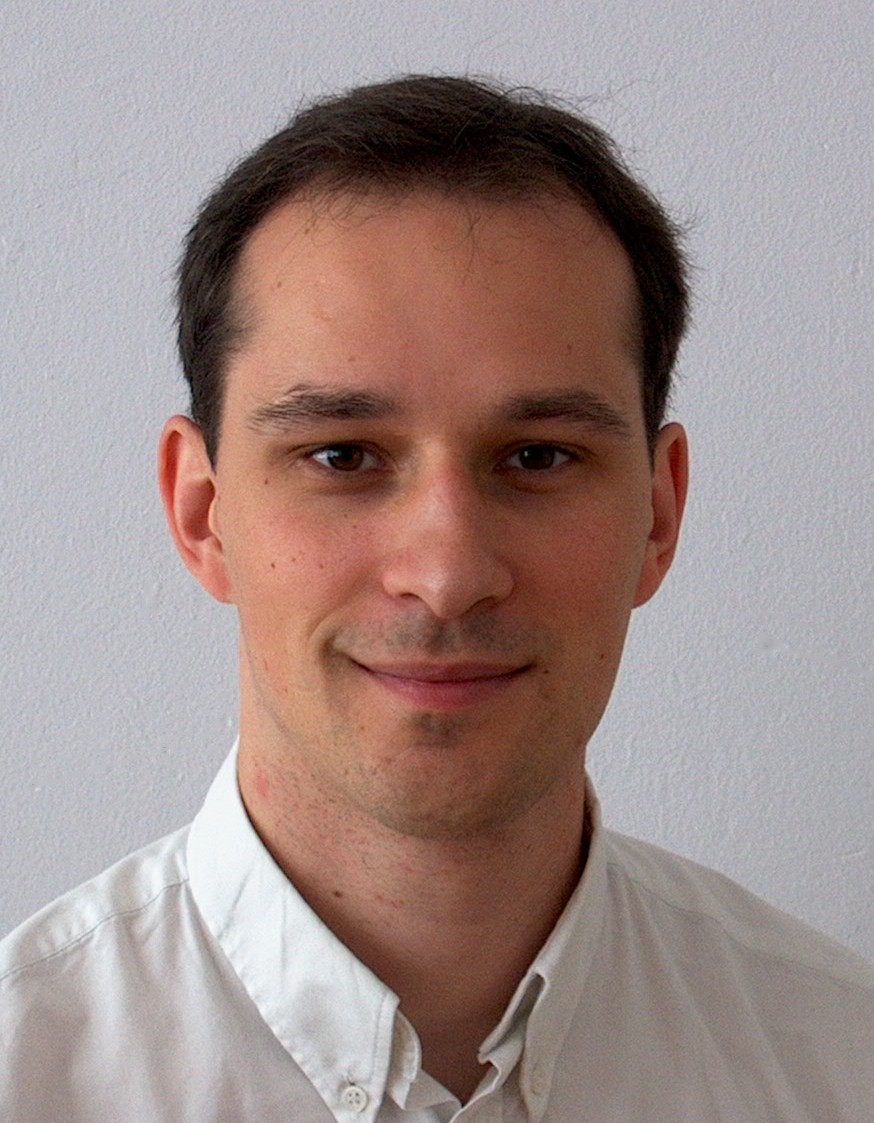}}]{Jérôme Louveaux}
(Fellow, IEEE) received the electrical engineering
degree and the Ph. D. degree from
the Université catholique de Louvain (UCLouvain),
Louvain-la-Neuve, Belgium in 1996 and 2000
respectively. From 2000 to 2001, he was a visiting
scholar in the Electrical Engineering department
at Stanford University, CA. From 2004
to 2005, he was a postdoctoral researcher at
the Delft University of Technology, Netherlands.
Since 2006, he has been a Professor in the
ICTEAM institute at UCL. His research interests
are in signal processing for digital communications, and in particular:
multicarrier modulations, xDSL systems, resource allocation, synchronization
and estimation.
\end{IEEEbiography}

\begin{IEEEbiography}[{\includegraphics[width=1in,height=1.25in,clip,keepaspectratio]{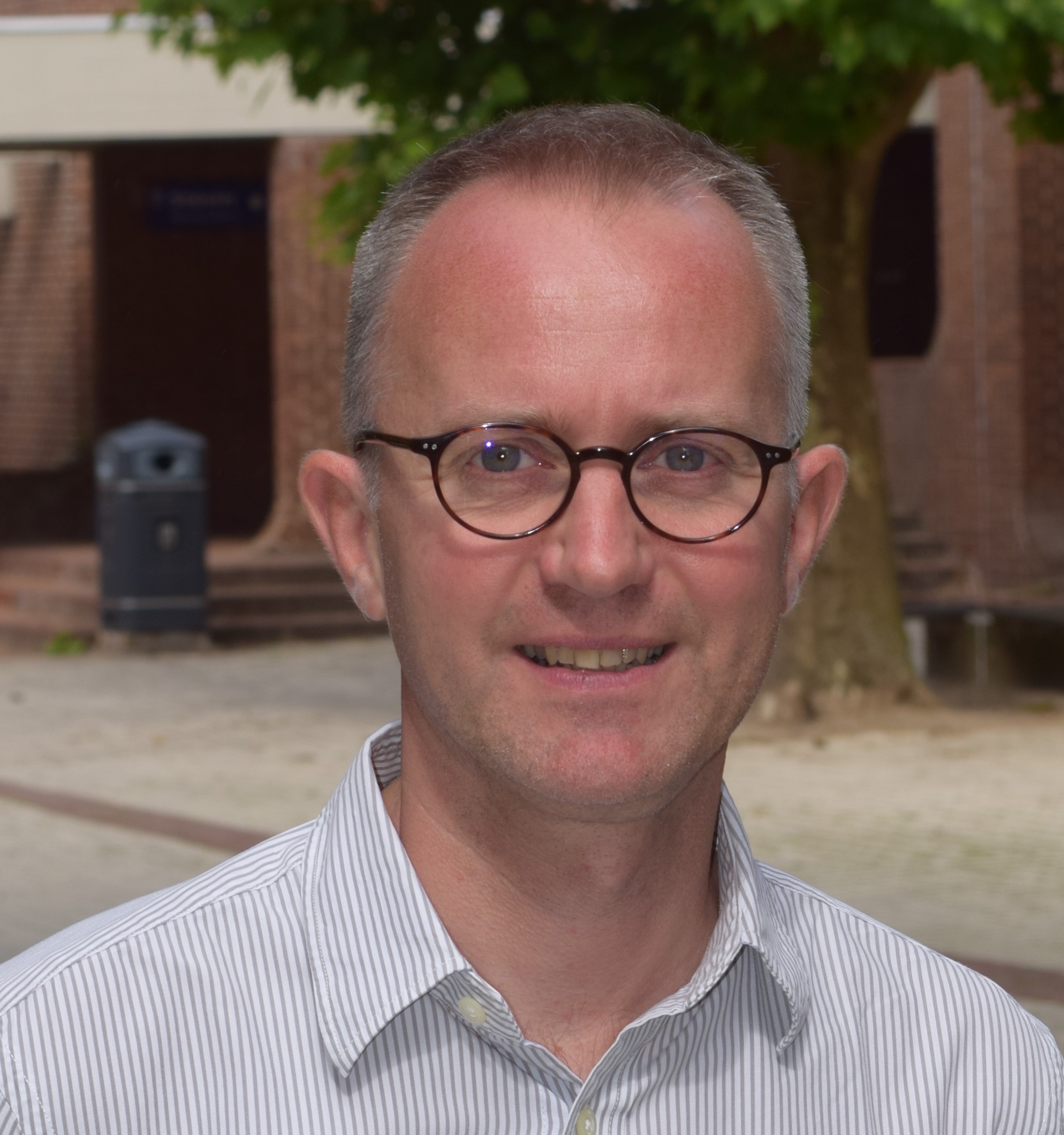}}]{Claude Oestges}
(Fellow, IEEE) received the M.Sc.
and Ph.D. degrees in electrical engineering from
the Université catholique de Louvain (UCLouvain),
Louvain-la-Neuve, Belgium, respectively, in 1996
and 2000. In January 2001, he joined as a Postdoctoral
Scholar the Smart Antennas Research Group (Information Systems Laboratory), Stanford University,
Stanford, CA, USA. From January 2002 to September
2005, he was associated with the Microwave Laboratory UCLouvain as a Post-doctoral Fellow of the
Belgian Fonds de la Recherche Scientifique (FRSFNRS). Claude Oestges is currently Full Professor with the Electrical Engineering Department, Institute for Information and Communication Technologies,
Electronics and Applied Mathematics (ICTEAM), UCLouvain. He was the
Chair of COST Action CA15104 IRACON (2016-2020). He is the author or
co-author of three books and more than 200 journal papers and conference
communications. He was the recipient of the 1999–2000 IET Marconi Premium
Award and of the IEEE Vehicular Technology Society Neal Shepherd Award in
2004 and 2012.
\end{IEEEbiography}

\begin{IEEEbiography}[{\includegraphics[width=1in,height=1.25in,clip,keepaspectratio]{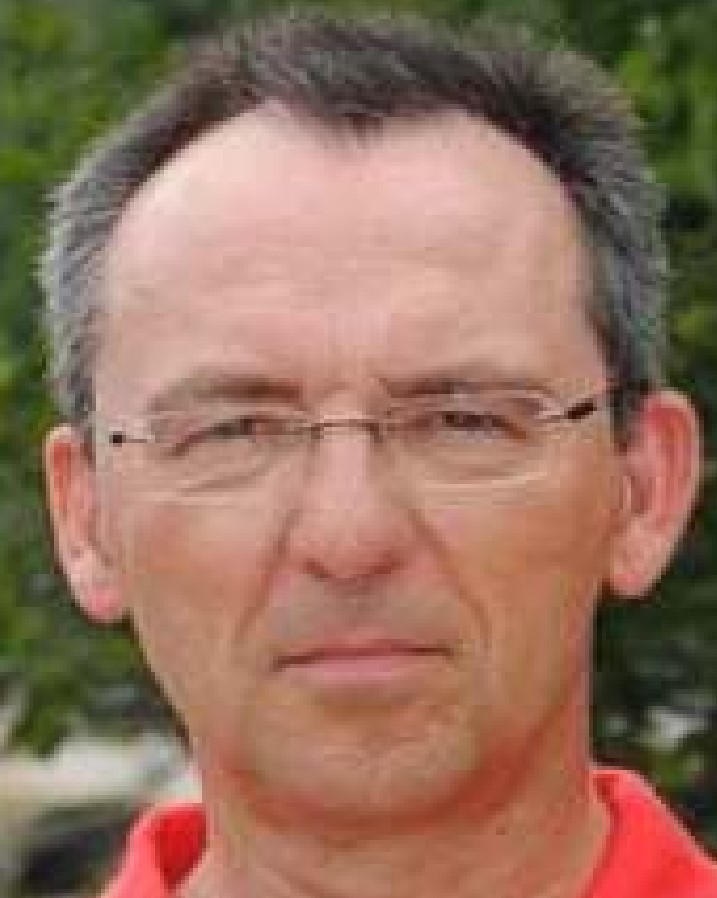}}]{Luc Vandendorpe}
(Fellow, IEEE) was born in
Mouscron, Belgium, in 1962. He received the degree
(summa cum laude) in electrical engineering and the
Ph.D. degree in applied science from the Catholic
University of Louvain (UCLouvain), Louvain La
Neuve, Belgium, in 1985 and 1991, respectively.
Since 1985, he has been with the Communications
and Remote Sensing Laboratory, UCLouvain, where
he first worked in the field of bit rate reduction
techniques for video coding. In 1992, he was a
Visiting Scientist and a Research Fellow with the
Telecommunications and Traffic Control Systems Group, Delft Technical
University, The Netherlands, where he worked on spread spectrum techniques
for personal communications systems. From October 1992 to August 1997,
he was a Senior Research Associate with the Belgian NSF, UCLouvain,
and an invited Assistant Professor. He is currently a Full Professor with
the Institute for Information and Communication Technologies, Electronics,
and Applied Mathematics, UCLouvain. His research interests are radar,
joint communication and radar systems, distributed massive MIMO, cell-free
networks, energy optimisation in IoT, game theory and localisation.
He is or has been a TPC Member of numerous IEEE conferences (VTC,
GLOBECOM, SPAWC, ICC, PIMRC, and WCNC). He was an Elected
Member of the Signal Processing for Communications Committee from 2000
to 2005 and the Sensor Array and Multichannel Signal Processing Committee
of the Signal Processing Society from 2006 to 2008 and from 2009 to 2011.
He was the Chair of the IEEE Benelux Joint Chapter on communications
and vehicular technology from 1999 to 2003. He was a Co-Technical Chair
for IEEE ICASSP 2006. He served as the Editor for synchronization and
equalization of IEEE Transactions on Communications from 2000 to 2002
and an Associate Editor for IEEE Transactions on Wireless Communications
from 2003 to 2005 and IEEE Transactions on Signal Processing from 2004
to 2006.
\end{IEEEbiography}

\vfill

\end{document}